\documentclass[journal,onecolumn]{IEEEtran}
\pdfoutput=1 %This ensures that the output file is a PDF (and not DVI) by setting a flag for pdflatex 
\usepackage{dsfont} 	%Package which allows double stroke symbols, used to write identity \id
\usepackage{comment}
\usepackage[numbers,sort&compress]{natbib} 	%Natbib is important for organising references
\usepackage[utf8]{inputenc} 	%Input what you want e.g., é, ?, a, ü
\usepackage[T1]{fontenc} 	%Output what you want e.g., é, ?, a, ü
\usepackage[british]{babel} 	%Do hyphenation according to British English
\usepackage[dvipsnames,x11names]{xcolor} 	%Add more colours (e.g., dark green)
\usepackage{amsmath,amssymb,amsthm,bm,amsfonts,bbm} 	%Usefull math packages without mathtools (to avoid conflict)
\usepackage{mathtools} 	%Usefull math package
\usepackage[colorlinks=true,citecolor=green,linkcolor=Purple,urlcolor=Magenta]{hyperref} 	%Hyperlinks with nice colours
\usepackage{marginnote}
\usepackage{ulem}
\newcommand{\ket}[1]{\vert#1\rangle}
\newcommand{\dket}[1]{\vert#1\rangle\hspace{-.8mm}\rangle}

\newcommand{\bra}[1]{\langle#1\vert}
\newcommand{\dbra}[1]{\langle\hspace{-.8mm}\langle #1\vert}

\newcommand{\ketbra}[2]{\vert #1 \rangle \hspace{-.4mm} \langle #2 \vert}
\newcommand{\dketbra}[2]{\vert #1 \rangle \hspace{-.8mm} \rangle \hspace{-.4mm} \langle\hspace{-.8mm}\langle #2 \vert}

\newcommand{\id}{\mathds{1}}

\DeclareMathOperator{\tr}{tr}
\DeclareMathOperator{\sgn}{sgn}
\newcommand*\dif{\mathop{}\!\mathrm{d}}
\newcommand{\mean}[1]{\left\langle#1\right\rangle}
\newcommand{\<}{\langle}
\renewcommand{\>}{\rangle}

\newcommand{\ie}{\textit{i.e.}}

\renewcommand{\H}{\mathcal{H}}

\renewcommand{\L}{\mathcal{L}}

\newcommand{\SU}{\mathcal{SU}}

\makeatletter 
\newcommand{\doublewidetilde}[1]{{%
  \mathpalette\double@widetilde{#1}%
}}
\newcommand{\double@widetilde}[2]{%
  \sbox\z@{$\m@th#1\widetilde{#2}$}%
  \ht\z@=.9\ht\z@
  \widetilde{\box\z@}%
}
\newcommand{\map}[1]{\widetilde{#1}}  

\makeatother
\newtheorem{definition}{Definition}
\newtheorem{theorem}{Theorem}
\newtheorem{lemma}{Lemma}
\newtheorem{notation}{Notation}
\newtheorem{proposition}{Proposition}
\hypersetup{pdftitle={{Optimal universal quantum circuits for unitary complex conjugation}}} %Sets the PDF title

\begin{document}
\author{Daniel Ebler, Michał Horodecki, Marcin Marciniak, Tomasz Młynik, Marco Túlio Quintino, Michał Studziński
\thanks{ %\hspace{-3.77mm} 
%\hspace{-2.8mm}
Daniel Ebler is at Theory Lab, Central Research Institute, 2012 Labs, Huawei Technology Co. Ltd, Hong Kong Science Park, Hong Kong SAR. %\\
Marcin Marciniak, Tomasz Młynik and  Michał Studziński are at Institute of Theoretical Physics and Astrophysics and National Quantum Information Centre in Gdańsk, and at Faculty of Mathematics, Physics and Informatics, University of Gdańsk. %\\ 
Michał Horodecki is at International Centre for Theory of Quantum Technologies, University of Gdańsk. %\\ 
Marco Túlio Quintino is at Sorbonne Université, CNRS, LIP6, Paris, France, and Institute for Quantum Optics and Quantum Information Vienna (IQOQI), Austrian Academy of Sciences, and at Faculty of Physics, University of Vienna.}}
	\date{\today}
    
\title{Optimal universal quantum circuits for unitary complex conjugation}

\maketitle
\begin{abstract}
Let $U_d$ be a unitary operator representing an arbitrary $d$-dimensional unitary quantum operation. This work presents optimal quantum circuits for transforming a number $k$ of calls of $U_d$ into its complex conjugate $\overline{U_d}$. Our circuits admit a parallel implementation and are proven to be optimal for any $k$ and $d$ with an average fidelity of $\left\langle{F}\right\rangle =\frac{k+1}{d(d-k)}$. Optimality is shown for average fidelity, robustness to noise, and other standard figures of merit. This extends previous works which considered the scenario of a single call ($k=1$) of the operation $U_d$, and the special case of $k=d-1$ calls. We then show that our results encompass optimal transformations from $k$ calls of $U_d$ to $f(U_d)$ for any arbitrary homomorphism $f$ from the group of $d$-dimensional unitary operators to itself, since complex conjugation is the only non-trivial automorphism on the group of unitary operators. Finally, we apply our optimal complex conjugation implementation to design a probabilistic circuit for reversing arbitrary quantum evolutions.
\end{abstract}
\section{Introduction}

The field of quantum information and computation typically describes how physical devices process information encoded into quantum systems. Such devices are often referred to as quantum gates in mathematical frameworks, and can be combined into quantum circuits by defining the order of processing \cite{NilsenChuangBook}. Conventionally, quantum gates were though of as operations which take quantum states as inputs, and, after processing, output a quantum state. However, in a paradigm sometimes referred to as ``higher-order quantum transformations'' quantum gates rather than quantum states may be subject to transformation~\cite{bisio19higher,perinotti16higher}. Such higher-order quantum transformations can be viewed as a ``circuit-board" approach for quantum operations, and has been found to have versatile applications encompassing quantum  circuit design~\cite{chiribella07,chiribella08,chiribella09networks}, control of quantum system dynamics~\cite{navascues17,trillo19,trillo22}, learning/storing of unitary operations~\cite{bisio10,sedlak18,mo2019quantum},  identifying cause-effect relations~\cite{chiribella18cause_effect}, and estimating and discriminating quantum channels~\cite{chiribella08memory_effects,bavaresco21,bavaresco21b}, among others.

When considering transformations between unitary operations, previous works have considered tasks such as unitary cloning~\cite{chiribella08clone}, unitary inversion~\cite{ebler16,quintino19PRL} (see also \cite{yoshida21}), unitary transposition~\cite{quintino19PRL}, unitary conjugation~\cite{miyazaki17}, unitary iteration~\cite{soleimanifar16}, unitary controlisation~\cite{araujo14control,bisio15control,gavorova20control,dong19control}, and general homomorphisms and anti-homomorphisms between unitary operators~\cite{bisio13,quintino21unitary}.
Unitary complex conjugation was first considered in Ref.~\cite{ebler16}, where non-exact optimal transformations between a single call of an arbitrary $d$-dimensional unitary operation into its complex conjugation were obtained. Later, Ref.~\cite{miyazaki17} presented a parallel quantum circuit which attains exact unitary conjugation by making $k=d-1$ calls of a $d$-dimensional unitary. Also, Ref.~\cite{quintino19PRA} considered the problem of exact unitary conjugation with an allowed rate of failure and showed that the probability of success is exactly zero when $k<d-1$. Hence, the probability of transforming $U_d$ into its complex conjugation cannot be enhanced by multi-call protocols such as the success-or-draw method~\cite{dong20}.  

In this work, we present deterministic and universal quantum circuits that transform $k$ calls of an arbitrary $d$-dimensional unitary operator $U_d$ into its complex conjugate $\overline{U_d}$. We present an explicit parallel circuit which makes no use of extra auxiliary space or memory for this task. We then prove that our construction attains the maximal performance allowed by quantum theory for any number of calls $k$ and any dimension $d$. 

\begin{figure}[!ht]
\begin{center}
	\includegraphics[width=0.49\textwidth]{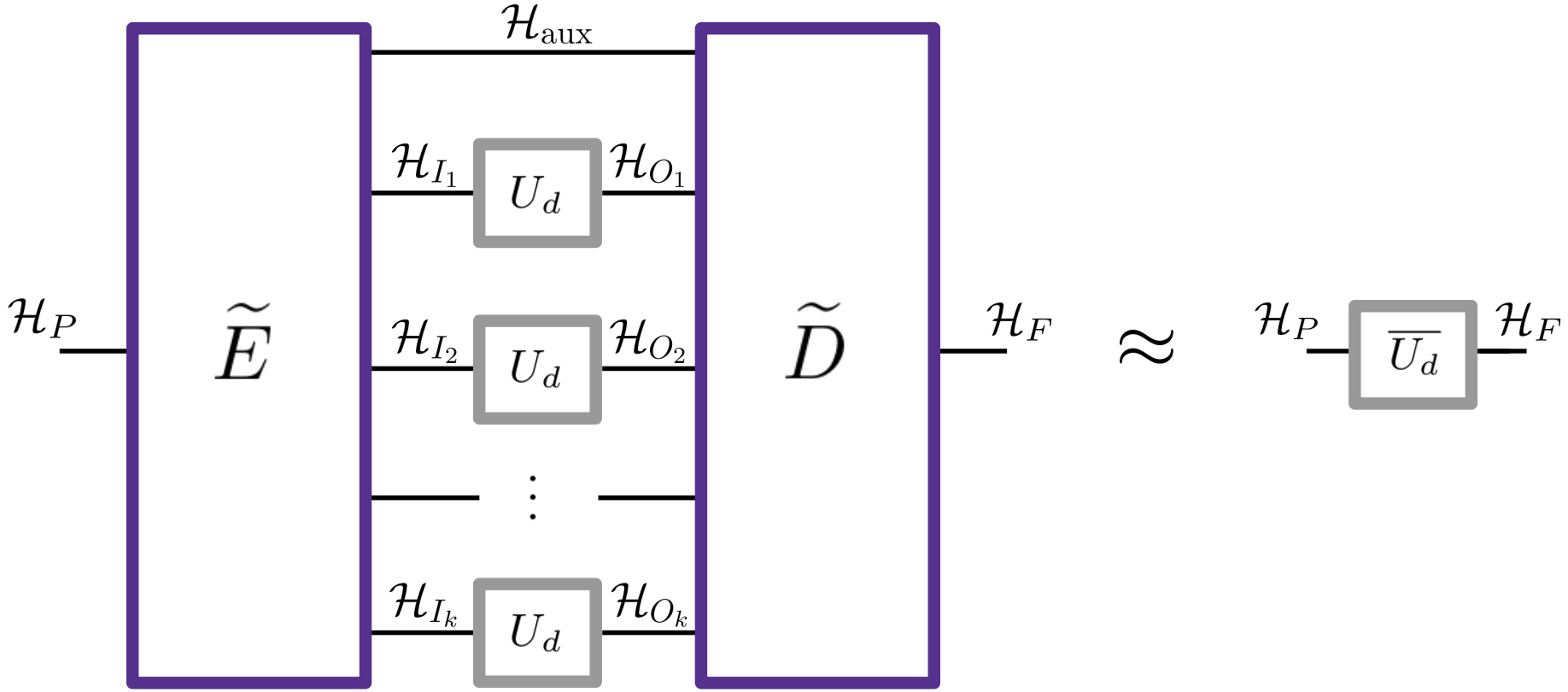}
\end{center}
\caption{Parallel superchannel transforming $k$ calls of a unitary operation $U_d$ into its complex conjugation $\overline{U_d}$. Here $\map{E}$ stands for an ``encoder'' quantum channel which is performed before the input operations $U_d^{\otimes k}$ and $\map{D}$ stands for a ``decoder'' quantum channel acting after the input operations. When $k<d-1$, unitary complex conjugation can only be obtained approximately with an average fidelity of $\mean{F}=\frac{k+1}{d(d-k)}$. Here we prove this number and present optimal quantum circuits for this task. The optimal circuit can operate in parallel on input transformations and does not require additional auxiliary spaces (depicted as the extra line labelled by the space $\mathcal{H}_{\rm aux}$ above the gates $U_d$).} \label{Fig1}
\end{figure}

\section{Transforming unitary operations with quantum circuits} %\label{sec:construction}
In the following, we formally introduce the problem considered in this manuscript.

We start by pointing out that, with no loss in performance, we can restrict circuits designed for unitary complex conjugation to  those that admit a parallel implementation, as depicted in Fig. \ref{Fig1}. Indeed, it was shown in Ref.~\cite{bisio13} that transformations $f$ acting on unitary operators $U$ as $f(U)$ can be implemented optimally by parallel circuits when $f$ is a homomorphism with respect to operator composition, \textit{i.e.,} for $f$ satisfying $f(UV)=f(U)f(V)$ for every unitary $U$ and $V$\footnote{Interestingly, parallel circuits are optimal even when considering transformations which do not respect a definite causal order, such as the quantum switch~\cite{chiribella09_switch} or general process matrices \cite{oreshkov11}.}. The case of unitary complex conjugation is naturally included, since $\overline{UV}=\overline{U}\; \overline{V}$.

In quantum theory, states are described by unit-trace, positive semidefinite linear operators\footnote{The symbol $\H$ denotes a finite-dimensional Hilbert space which is isomorphic to $\mathbb{C}_d$ for some dimension $d\in\mathbb{N}$. The symbol $\L(\H)$ denotes the set of linear operators acting on $\H$ (endomorphisms).} -- \textit{i.e.}, $\rho\in\L(\H)$, $\rho\geq0$, and $\tr(\rho)=1$. Deterministic transformations between input-states and output-states are implemented by quantum channels, which are completely positive\footnote{A linear map $\map{C}:\L(\H_{\bm{I}})\to\L(\H_{\bm{O}})$ is positive when $\map{C}(\rho)\geq0$ for every positive semidefinite linear operator $ \rho\geq0, \rho\in \L(\H_{\bm{I}})$. \\
A map is CP if it is positive for all its trivial extensions, that is, $\map{C}\otimes\map{\id}(\sigma)\geq 0$ for every $\sigma\geq0, \sigma\in\L(\H\otimes\H_\text{aux})$ where $\map{\id}:\L(\H_\text{aux})\to\L(\H_\text{aux})$ is the identity map on an arbitrary finite-dimensional space $\H_\text{aux}$, \textit{i.e.,} $\map{\id}(\rho)=\rho, \forall \rho\in\L(\H_\text{aux})$.}  (CP) and trace-preserving\footnote{A linear map $\map{C}:\L(\H_{\bm{I}})\to\L(\H_{\bm{O}})$ is TP when $\tr(\map{C}(\rho))=\tr(\rho)$ for every linear operator $\rho\in \L(\H_{\bm{I}})$.} (TP) linear maps\footnote{In this paper, we reserve the word map for linear transformations between linear operators. Also, adopt the convention of writing a tilde on the top of maps to make an easy distinction between operators, \textit{e.g,} $C\in\L(\H_{\bm{I}}\otimes \H_{\bm{O}})$ and maps $\map{C}:\L(\H_{\bm{I}})\to\L(\H_{\bm{O}})$. This convention is particularly useful when dealing with Choi representations of linear maps.} $\map{C}:\L(\H_{\bm{I}})\to\L(\H_{\bm{O}})$, where $\H_{\bm{I}}$ and $\H_{\bm{O}}$ are the linear spaces relative to input and output, respectively.  A quantum channel $\map{U}:\L(\H_{\bm{I}})\to\L(\H_{\bm{O}})$ is said to be unitary if it can be written as $\map{U}(\rho)=U_d \rho U_d^\dagger$ for some $d$-dimensional unitary operator $U_d$. Note that unitary operators which only differ by a global phase represent the same unitary channel. Hence, with no loss of generality, we may assume that the unitary operator $U_d$ belongs to the group of special unitary operators $\SU(d)$, that is, unitary operators with determinant one.

Let $\map{C_\text{in}}:\L(\H_{\bm{I}})\to\L(\H_{\bm{O}})$ be an arbitrary input channel to a quantum circuit, which transforms it into an output channel $\map{C_\text{out}}:\L(\H_P)\to\L(\H_F)$. The labels $P$ and $F$ stands for past and future, respectively (see Fig. \ref{Fig1}). Quantum circuits designed to obtain the transformation $\map{C_\text{in}}\mapsto \map{C_\text{out}}$ may be analysed by means of encoder and decoder channels \cite{chiribella07,chiribella08}, a method which we describe in the following.

\begin{enumerate}
	\item Before performing the input operation $\map{C_\text{in}}$, we apply a quantum channel (encoder) 
$\map{E}:\L(\H_P)\to\L(\H_{\bm{I}}\otimes\H_\text{aux})$, where $\H_\text{aux}$ is an arbitrary auxiliary (memory) space which may be used to improve performance.
\item Then, the operation $\map{C_\text{in}}\otimes \map{\id_\text{aux}}$ is applied, where $\map{\id_\text{aux}}$ is the identity map, which acts trivially on the auxiliary system, \textit{i.e.,} $\map{\id_\text{aux}}(\rho)=\rho$, for every operator $\rho$.
\item Finally, we perform a quantum channel (decoder) ${\map{D}:\L(\H_{\bm{O}}\otimes\H_\text{aux})\to\L(\H_F)}$ to obtain the output
\begin{align} \label{eq:superchannel}
\map{C_\text{out}}=
\map{D} \circ \Big(\map{C_\text{in}}\otimes \map{\id_\text{aux}}\Big)\circ \map{E}.
\end{align} 
\end{enumerate}

The use of $k$ independent calls of a unitary channel $\map{U_d}$ may be mathematically represented by a single channel $\map{U_d}^{\otimes k}$. Then, by identifying $\map{U_d}^{\otimes k}$ as $\map{C_\text{in}}$ in the routine described above, a parallel quantum circuit transforms $k$ calls of $\map{U_d}$ as 
\begin{align} \label{eq:parallelED}
	 \map{U_d}^{\otimes k} \mapsto& \map{D} \circ \Big(\map{U_d}^{\otimes k}\otimes \map{\id}\Big)\circ \map{E},
\end{align} 
as illustrated in Fig.~\ref{Fig1}. When the input channel is composed of $k$ operations, it is convenient to write the input space $\H_{\bm{I}}$ explicitly as 
$\H_{\bm{I}}:=\H_{I_1}\otimes \ldots \otimes \H_{I_{k}}$, where the bold letter $\bm{I}$ emphasises that the linear space $\H_{\bm{I}}$ has an  internal tensor product structure (similarly, we write $\H_{\bm{O}}:=\H_{O_1}\otimes \ldots \otimes \H_{O_{k}}$). Further, we have
$\H_P\cong\H_F\cong\H_{I_i}\cong\H_{O_i}\cong \mathbb{C}_d$ for every $i\in\{1,\ldots,k\}$ and
$\H_{\bm{I}}\cong\H_{\mathbf{O}}\cong \mathbb{C}_d^{\otimes k}$ since the quantum circuit transforms $k$ calls of $d$-dimensional unitary operations into a $d$-dimensional unitary operation.

The performance of transformations can be quantified by means of channel fidelity, which is known to respect several  properties that are relevant for information processing \cite{raginsky01}. 
The fidelity between quantum channels is usually defined through their Choi operators. The Choi operator of a quantum channel is obtained via the Choi isomorphism is defined as follows. Let 
	 $T:\H_{\bm{I}}\to\H_{\bm{O}}$ be a linear operator. Its Choi vector $\dket{T}\in\H_{\bm{I}}\otimes\H_{\bm{O}}$ is, 
\begin{equation}
	\dket{T}:= \sum_{i} \ket{i} \otimes \Big(T \ket{i}\Big),
\end{equation}
where $\{\ket{i}\}_{i=1}^d$ is the canonical orthonormal basis for $\mathbb{C}_d$, often referred to as computational basis.
	Further, let
	 $\map{C}:\L(\H_{\bm{I}})\to\L(\H_{\bm{O}})$ be a linear map between linear operators. Its Choi operator $C\in\L(\H_{\bm{I}}\otimes\H_{\bm{O}})$ is
\begin{equation}
	C:= \sum_{ij} \ketbra{i}{j} \otimes \map{C}\Big( \ketbra{i}{j}\Big).
\end{equation}
Direct calculation shows that the Choi operator of a unitary channel $\map{U_d}(\rho)=U_d\rho U_d^\dagger$ may be conveniently expressed as $\dketbra{U}{U}$.
 
The fidelity between an arbitrary quantum channel $\map{C}$ and a unitary channel acting as $\map{U_d}(\rho)=U_d\rho U_d^\dagger$ on an input $\rho$ is given by \cite{raginsky01}
%\begin{equation}
%    F(\map{C},\map{D}):=\frac{1}{d^2} \left(  \tr \left[ \sqrt{\sqrt{C}D \sqrt{C}}  %\right]   \right)^2 \ ,
%\end{equation}
%where $C$ and $D$ are the Choi operators of $\map{C}$ and $\map{D}$. The expression simplifies in the case one channel is unitary channel (e.g. $\map{D}$ acting as $\map{U_d}(\rho)=U_d\rho U_d^\dagger$ on an input state $\rho$) to
\begin{equation}\label{eq:channel_fidelity}
	F(\map{C},\map{U_d}):=\frac{1}{d^2}\dbra{U_d} C \dket{U_d},
\end{equation}
where $\dket{U_d}$ the Choi vector of $U_d$.

We now turn to the task of transforming $k$ copies of a unitary operator $U_d$ into its complex conjugate $\overline{U_d}$  with respect to the canonical computational basis $\{\ket{i}\}_{i=1}^d$.

    Let us decompose $U_d=\sum_{ij} u_{ij}\ketbra{i}{j}$ in the computational basis. The complex conjugate of $U_d$ in said basis reads $\overline{U_d}:=\sum_{ij} \overline{u_{ij}} \ketbra{i}{j}$. 
Note that the definition of $\overline{U_d}$ depends on the choice of basis. To see this, let us express $U_d$ in a different orthonormal basis $\{\ket{i'}\}_{i=1}^d$ as $U_d=\sum_{ij} u_{ij}'\ketbra{i'}{j'}$. Since $u_{ij}'=\bra{i'}U_d\ket{j'}$ may be different to $u_{ij}=\bra{i}U_d\ket{j}$, we conclude
\begin{equation}
    \overline{U_d}:=\sum_{ij} \overline{u_{ij}} \ketbra{i}{j}\neq  \sum_{ij} \overline{u_{ij}'} \ketbra{i'}{j'}=:{\overline{U_d}}'. 
\end{equation}
However, $\overline{U_d}$ and $\overline{U_d}'$ are unitarily equivalent: let $W\in\L(\mathbb{C}_d)$ be a unitary operator mapping the $\{\ket{i'}\}_{i=1}^d$  basis into the $\{\ket{i}\}_{i=1}^d$ basis. One simple example is given by $W:=\sum_i \ketbra{i}{i'}$. Then, we find $WU_dW^\dagger=\sum_{ij}u_{ij}'\ketbra{i}{j}$, and  
\begin{align}
   \overline{WU_dW^\dagger}=&   \sum_{ij} \overline{u_{ij}'} \ketbra{i'}{j'} \\
   =&{\overline{U_d}}',
\end{align}
leading to the identity $\overline{U_d}'=\overline{W}\;\overline{U_d}\; W^T$, where $\overline{W}^\dagger=W^T$ is the transposition of $W$ on the computational basis\footnote{Let $A=\sum_{ij}a_{ij}\ketbra{i}{j}$ be an arbitrary operator acting on $\mathbb{C}_d$,  its transposition in the computational basis is defined as $A^T=\sum_{ij}a_{ij}\ketbra{j}{i}$.}.

%and $\bra{i'}=W\bra{i}$, $i=1,\ldots,d$, for some unitary $W$, the one can easily check that $\overline{U_d}'=(WW^T)\overline{U_d}(WW^T)^\dagger$, where $\overline{U_d}'$ stands for the operator conjugate to $U_d$ with respect to the basis $\{\bra{i'}\}_{i=1}^d$, and $W^T$ is the operator transposed to $W$ with respect to the computational basis.

As discussed earlier, since the complex conjugation function is a homomorphism, \ie, $\overline{UV}=\overline{U}\,\overline{V}$, we can restrict our analysis to parallel quantum circuits. That is, we focus on circuits which can be as,
\begin{align}
    \map{U_d}^{\otimes k} \mapsto& \map{D} \circ \Big(\map{U_d}^{\otimes k}\otimes \map{\id}\Big)\circ \map{E}\approx \map{\overline{U_d}},
\end{align}
where $\map{E}$ and $\map{D}$ are quantum channels, as presented in Eq.~\eqref{eq:superchannel}. We now set the average channel fidelity as our figure of merit. More precisely, we quantify the performance of our transformation via
\begin{equation}
	\mean{F}:=\int_{U_d\in \SU(d)}  F\Big(\map{D} \circ \big(\map{U_d}^{\otimes k}\otimes \map{\id}\big)\circ \map{E}	, \, \map{\overline{U_d}}\Big)\,\dif U_d ,
\end{equation}
where the fidelity of quantum channels is defined in Eq.~\eqref{eq:channel_fidelity} and the integral is taken in accordance to the Haar measure.
As proven in Ref.~\cite{quintino21unitary}, for the case of homomorphic transformations (as the complex conjugation), the optimal average fidelity coincides with the optimal worst-case fidelity and there is a one-to-one correspondence between average fidelity and the robustness of the transformation to noise, ensuring that the average fidelity is indeed a relevant figure of merit.

\section{Optimal quantum circuit for unitary complex conjugation} \label{sec:construction}
In this section, we state our main result and provide an explicit construction for the optimal quantum circuit for complex conjugation of an arbitrary unitary operation. The proof of optimality and performance evaluation of the circuit presented here are given in Sec.~\ref{sec:proof}.
\begin{theorem}
\label{thm:main}
	Let $U_d$ be a unitary operator representing an arbitrary $d$-dimensional unitary operation $\map{U_d}(\rho)=U_d\rho U_d^\dagger$. When $k\leq d-1$ uses are available, the optimal quantum circuit which transforms $k$ uses of $U_d$ into its complex conjugation $\overline{U_d}$ attains average fidelity $\mean{F}=\frac{k+1}{d(d-k)}$. \\ Additionally, the optimal circuit admits a parallel implementation which makes no use of additional auxiliary spaces.
\end{theorem}
Note that when $k=d-1$, we obtain $\mean{F}=1$, hence unitary complex conjugation may be done exactly when the number of calls respects $k\geq d-1$. This optimal and exact transformation for the $k=d-1$ case was first presented at Ref.~\cite{miyazaki17} and will be reviewed in this section for convenience.
	
It is useful to state some facts regarding the antisymmetric subspace, which will play a major role in constructing optimal quantum circuits for unitary complex conjugation. Let $S_N$ be the group of permutations of $N$ elements, and $V(\pi)\in \mathbb{C}_d^{\otimes N}$ be a representation for the permutation $\pi\in S_N$ in $\mathbb{C}_d^{\otimes N}$, explicitly defined via
\small
\begin{equation}
	V(\pi)\ket{\psi_1}\otimes\ket{\psi_2}\otimes\ldots\ket{\psi_N}=\ket{\psi_{\pi^{-1}(1)}}\otimes\ket{\psi_{\pi^{-1}(2)}}\otimes\ldots\ket{\psi_{\pi^{-1}(N)}} .
\end{equation}
\normalsize
	The antisymmetric subspace of $N$ $d$-dimensional parties is a subspace of $\mathbb{C}_d^{\otimes N}$ formed by vectors $\ket{\psi}\in\mathbb{C}_d^{\otimes N}$ satisfying
\begin{equation}
	V(\pi)\ket{\psi}=\sgn(\pi)\ket{\psi}, \quad \forall \pi\in S_N
\end{equation}
	where $\sgn(\pi)$ is the sign\footnote{That is, $\sgn(\pi)=+1$ when $\pi$ can be written as a composition of even transpositions (pairwise permutation) and $\sgn(\pi)=-1$ when $\pi$ can be written as a composition of odd transpositions.} of the permutation $\pi$. When $d=N$, the antisymmetric subspace of $\mathbb{C}_d^{\otimes N}$ is one-dimensional, and it is spanned by a single vector, known as the totally antisymmetric state,
\begin{align}
	\ket{\psi_{A_d}}:=\frac{1}{\sqrt{N!}}\sum_{\pi\in S_N} \sgn(\pi)  \ket{\pi(1)}\otimes\ket{\pi(2)}\otimes\ldots\otimes\ket{\pi(d)},
\end{align}
where $\{\ket{1},\ket{2},\ldots,\ket{d}\}$ is the computational basis  for $\mathbb{C}_d$. When $N=d$, the totally antisymmetric state is known to satisfy 
\begin{equation}
U_d^{\otimes d}\ket{\psi_{A_d}}=\ket{\psi_{A_d}}, \quad \forall U_d\in\SU(d),
\end{equation}
which is equivalent to
\begin{equation} \label{eq:total_anti}
\Big(\id_d\otimes U_d^{\otimes (d-1)}\Big) \ket{\psi_{A_d}}=\Big(U_d^\dagger \otimes \id_d^{\otimes(d-1)} \Big) \ket{\psi_{A_d}}.%, \quad \forall U_d\in\SU(d).
\end{equation}

The transposition vector $\ket{\psi}\in\mathbb{C}_d$ on the computational basis is a linear map from $\mathbb{C}_d$ to its dual space defined via $\ket{i}^T:=\bra{i}$. One can then verify that $\big(U_d\ket{\psi}\big)^T=\bra{\psi}\;U_d^T$. When $\ket{\psi}_{12}\in\mathbb{C}_d\otimes\mathbb{C}_{d'}$ is a bipartite vector which can be written as $\ket{\psi}_{12}=\sum_{ij} \gamma_{ij}\ket{i}_1\otimes\ket{j}_2$ in the computational basis, the partial transpose on the first linear space reads
\begin{align}
    \ket{\psi}_{12}^{T_1}:&=\sum_{ij} \gamma_{ij}\ket{i}_1^T\otimes \ket{j}_2 \\
    &=\sum_{ij} \gamma_{ij}\bra{i}_1\otimes \ket{j}_2 . 
\end{align}

Hence, doing the partial transposition on the first linear space of Eq.~\eqref{eq:total_anti},  we obtain
\begin{align} \label{eq:total_anti2}
\id_d\otimes U_d^{\otimes (d-1)}  \ket{\psi_{A_d}}^{T_1} &=\Big(\id_d\otimes U_d^{\otimes (d-1)}  \ket{\psi_{A_d}}\Big)^{T_1} \\
&=\Big(U_d^\dagger \otimes \id_d^{\otimes(d-1)} \; \ket{\psi_{A_d}}\Big)^{T_1} \\
&=\ket{\psi_{A_d}}^{T_1} \Big(\overline{U_d}\otimes\id_d^{\otimes (d-1)} \Big).
\end{align}
	With Eq.~\eqref{eq:total_anti2}, we can now recover the result of Ref.~\cite{miyazaki17} stating that when $k=d-1$, it is possible to transform $U_d^{\otimes (d-1)}$ into $\overline{U_d}$ exactly. We now define the linear transformation
	$V:\mathbb{C}_d\to\mathbb{C}_d^{\otimes (d-1)}$ as,
 \small
	%$V:\H_P \to \H_I$
\begin{align}
	V:=&\frac{\sqrt{d!}}{\sqrt{(d-1)!}}\ket{\psi_{A_d}}^{T_1} \\
%=&\frac{1}{\sqrt{(d-1))!}}\sum_{\pi\in S_N} \bra{\pi(1)}_{P}\otimes\ket{\pi(2)}_{I_1}\otimes\ldots\otimes\ket{\pi(d)}_{I_k},
	=&\frac{1}{\sqrt{(d-1))!}}\sum_{\pi\in S_N}\sgn(\pi)\, \bra{\pi(1)}\otimes\ket{\pi(2)}\otimes\ldots\otimes\ket{\pi(d)}  .
\end{align}
\normalsize
Note that $V$ is an isometry since $V^\dagger V = \id_d$. Now, in order to transform $k=d-1$ uses of $U_d$ into $\overline{U_d}$, we simply perform  $V$ before $U_d^{\otimes (d-1)}$, and the operator $V^\dagger$ after to obtain
\begin{align}
	V^\dagger \; U_d^{\otimes (d-1)} \; V=\overline{U_d}.
\end{align}

	In the following, we move to the case $k<d-1$. The identification of the optimal circuit for this scenario is the main finding of this work.
	The case $k<d-1$ turns out to be more involved, as we need to consider scenarios where the antisymmetric subspace is not one-dimensional. When $k+1<d$, the antisymmetric subspace of $\mathbb{C}_d^{\otimes (k+1)}$ is spanned by a set of $\binom{d}{k+1}:=\frac{d!}{(k+1)!(d-k-1)!}$ orthogonal vectors given by $\ket{\psi^{(i)}_{A_{d,k+1}}}$. These vectors may be explicitly constructed by a simple method: choose $k+1$ vectors from $\{\ket{i}\}_{i=1}^d$ and add the permutation-signed sum over all $k+1$ permutations. For the sake of concreteness we set the vector with index $i=1$ as 
\begin{equation}\begin{split}
\ket{&\psi^{(1)}_{A_{d,k+1}}}:=\frac{1}{\sqrt{(k+1)!}} \\ 
&\sum_{\pi\in S_{k+1}} \sgn(\pi) \ket{\pi(1)}\otimes\ket{\pi(2)}\otimes\ldots\otimes\ket{\pi(k+1)}.
\end{split}
\end{equation}
	Since $k+1<d$, other orthogonal antisymmetric vectors can be constructed -- for instance
 \small
\begin{equation}\begin{split}
\ket{&\psi^{(2)}_{A_{d,k+1}}}:=\\
&\frac{1}{\sqrt{(k+1)!}}\sum_{\pi\in S_{k+1}} \sgn(\pi) 
\ket{\pi(2)}\otimes\ket{\pi(3)}\otimes \ldots\otimes\ket{\pi(k+2)} \,
\end{split}
\end{equation}
\normalsize
which also respects $V(\pi)\ket{\psi^{(2)}_{A_{d,k+1}}}=\sgn(\pi)\ket{\psi^{(2)}_{A_{d,k+1}}}$ for every permutation $\pi\in S_{k+1}$.
	The set of $\binom{d}{k+1}$ vectors of this form then span the whole antisymmetric space of $\mathbb{C}_d^{\otimes (k+1)}$.
We can now write the projector onto the antisymmetric space as,	
\begin{align}
	A(d,k+1)=\sum_{i=1}^{\binom{d}{k+1}} \ketbra{\psi^{(i)}_{A_{d,k+1}}}{\psi^{(i)}_{A_{d,k+1}}},
\end{align}
which may also be equivalently written as \cite{harrisBook}
\begin{align}
	A(d,k+1)=\sum_{\pi\in S_{k+1}} \sgn({\pi}) V(\pi) .
\end{align}
We are now in position to present the optimal circuit for unitary complex conjugation. Our circuit does not make use of any auxiliary space. The action of the encoder channel %$\map{E}:\L(\mathbb{C}_d)\to\L(\mathbb{C}_d^{\otimes k})$
$\map{E}:\L(\H_P)\to\L(\H_{\bm{I}})$, with $\H_P\cong\mathbb{C}_d$, and $\H_{\bm{I}}\cong\mathbb{C}_d^{\otimes k}$
may be equivalently expressed in the Choi or operator-sum representation, which are respectively given by
\begin{align} \label{eq:E}
	\map{E}(\rho):=& \tr_P\Big((\rho^T_P \otimes \id_{\bm{I}}) \; E_{P\bm{I}} \Big) \\
				=& \sum_{i=1}^{\binom{d}{k+1}} K_i \rho K_i^\dagger ,
\end{align}
where $\tr_P$ denotes the trace over the past space $\mathcal{H}_P$, and $E\in\L(\H_P\otimes\H_{\bm{I}}),$
\begin{align}
    E:=\frac{d}{\binom{d}{k+1}} A(d,k+1)
\end{align} is the Choi operator of the map $\map{E}$ and 
\begin{align}
K_i:= \sqrt{\frac{d}{\binom{d}{k+1}}} \ket{\psi^{(i)}_{A_{d,k+1}}}_{P\bm{I}}^{T_P}
\end{align} are the Kraus operators of $\map{E}$.
The decoder channel 
%$\map{D}:\L(\mathbb{C}_d^{\otimes k})\to\L(\mathbb{C}_d)$ 
$\map{D}:\L(\H_{\bm{O}})\to\L(\H_F)$, where $\H_{\bm{O}}\cong\mathbb{C}_d^{\otimes k}$ and $\H_F\cong\mathbb{C}_d$,
may be equivalently expressed in the Choi  or operator-sum representation, given by
\begin{align} \label{eq:D}
	\map{D}(\rho):=& \tr_{\bm{O}}\Big((\rho^T_{\bm{O}} \otimes \id_F) \; D_{\bm{O}F} \Big) \\
=& \sum_{i=1}^{\binom{d}{k+1}} K_i^\dagger \rho_{\bm{O}} K_i + \tr\Big(\big(\id_d^{\otimes k}- A(d,k)\big)\;\rho\Big)\; \sigma_F .
\end{align}
\normalsize
Here,  $D\in\L(\H_{\bm{O}}\otimes\H_F)$,
\begin{equation}
	D:=\frac{\binom{d}{k}}{\binom{d}{k+1}} A(d,k+1)_{\bm{O}F}+\Big(\id_d^{\otimes k}- A(d,k)\Big)_{\bm{O}}\otimes \sigma_F,
\end{equation}
is the Choi operator of the map $\map{D}$
and $\sigma\in\L(\H_F)\cong\L(\mathbb{C}_d)$ is an arbitrary quantum state, \textit{i.e.,} $\sigma\geq0$ and $\tr(\sigma)=1$.

In Sec.~\ref{sec:proof}, we prove that the circuit with the encoder described in Eq.~\eqref{eq:E} and the decoder of Eq.~\eqref{eq:D} achieves optimal unitary complex conjugation. More precisely, we will show that for any unitary operation $\map{U}_d$ the fidelity between $\map{D}\circ \map{U}_d^{\otimes k} \circ \map{E}$ and $\map{\overline{U}}_d$ is given by
%\begin{equation}
%	\map{D}\circ \map{U_d}^{\otimes k} \circ \map{E} \approx \map{\overline{U_d}},
%\end{equation}
\begin{equation}
	F\Big(\map{D}\circ \map{U}_d^{\otimes k} \circ \map{E} , \map{\overline{U}}_d\Big)=\frac{k+1}{d(d-k)} .
\end{equation}
Moreover, we show that no other protocol can perform better: for every possible choice of encoder $\map{E}$ and decoder $\map{D}$ the average channel fidelity satisfies
\small
\begin{equation}
    \int_{U_d\in \SU(d)}  F\Big(\map{D} \circ \big(\map{U}_d^{\otimes k}\otimes \map{\id}\big)\circ \map{E}	, \, \map{\overline{U}}_d\Big)\,\dif U_d\leq \frac{k+1}{d(d-k)}.
\end{equation}
\normalsize
As discussed earlier, the optimal average fidelity coincides with the optimal worst-case fidelity since the complex conjugation is a homomorphism. Also, for homomorphic transformations, protocols which are optimal with respect to average fidelity are also optimal with respect to the robustness to white noise \cite{quintino21unitary}.

\section{Proof of optimality} \label{sec:proof}

\subsection{Quantum superchannels}

	A useful mathematical framework for analysing transformations between quantum operations is given by the formalism of quantum superchannels, also referred to as 1-slot quantum combs or quantum supermaps \cite{chiribella07,chiribella08,chiribella09networks}.

\begin{definition}
\label{def1}
	A linear operator $S\in\L(\H_P\otimes\H_{\bm{I}}\otimes\H_{\bm{O}}\otimes\H_F)$ is a parallel superchannel if it can be written as
	\begin{equation}
		S= \tr_\text{aux} \Big((E_{P\bm{I}\text{aux}}^{T_\text{aux}}\otimes \id_{\bm{O}F})\, (\id_{PI}\otimes D_{\text{aux}\bm{O}F})\Big),
	\end{equation}
	where $\H_\text{aux}$ is an arbitrary finite-dimensional linear space, $^{T_{\text{aux}}}$ is the partial transpose in the $\H_\text{aux}$ space, and $E\in \L(\H_P\otimes  \H_{\bm{I}}\otimes \H_{\text{aux}} )$ and 	$D\in \L(\H_{\text{aux}}\otimes\H_{\bm{O}}\otimes\H_F)$ are operators such that,
\begin{align}
		& E\geq0, \quad \tr_{{\bm{I}\text{aux}}}(E) = \id_{P} \\
   & D\geq0, \quad \tr_{F}(D) = \id_{\text{aux}\bm{O}}.
\end{align}
	That is, $E$ is the Choi operator of a quantum channel from $\L(\H_P)$ to $\L(\H_{\bm{I}}\otimes \H_\text{aux})$
	and $D$ is the Choi operator of a quantum channel from $\L(\H_\text{aux}\otimes\H_{\bm{O}})$ to $\L(\H_F)$.
\end{definition}

	The concept of superchannel presented in Def.~\ref{def1} is useful because it incorporates every possible quantum circuit based on encoder-decoder schemes~\cite{chiribella08,chiribella09networks}: a positive semidefinite operator $S\geq0$, $S\in\L(\H_P\otimes\H_{\bm{I}}\otimes\H_{\bm{O}}\otimes\H_F)$ is a parallel superchannel if and only if it there exists a positive semidefinite operator $C\geq0$ $C\in\L(\H_P\otimes\H_{\bm{I}})$ such that,
\begin{align} \label{eq:superchannel_characterised}
	\tr_{F}(S)&=C_{P{\bm{I}}}\otimes \id_{\bm{O}}, \\
	\tr_{\bm{I}}(C)&=\id_P
\end{align}
that is, $C$ is the Choi operator of a quantum channel from $\L(\H_P)$ to $\L(\H_{\bm{I}})$. The decomposition of Eq.~\eqref{eq:superchannel_characterised} can be used to explicitly construct an encoder and decoder channel \cite{chiribella08,chiribella09networks}. For that, we set the auxiliary space as $\H_\text{aux}:=\H_{P'}\otimes\H_{{\bm{I}}'}$ where $\H_{P'}$ and $\H_{{\bm{I}}'}$ are isomorphic $\H_{P}$ and $\H_{{\bm{I}}}$ and define 
\footnotesize
\begin{align}
E:=& \Big(\sqrt{C_{P{\bm{I}}}}\otimes\id_{P'{\bm{I}}'}\Big) \dketbra{\id}{\id}_{PP'}\otimes\dketbra{\id}{\id}_{{\bm{II}}'}\Big(\sqrt{C_{P{\bm{I}}}}\otimes\id_{P'{\bm{I}}'}\Big)^\dagger \\
D:=& \Big(\sqrt{C_{P'{\bm{I}}'}}^{-1}\otimes\id_{\bm{O}F}\Big) S_{P'{\bm{I}}'\bm{O}F}\Big(\sqrt{C_{P'{\bm{I}}'}}^{-1}\otimes\id_{\bm{O}F}\Big)^\dagger,
\end{align}
\normalsize
where $\sqrt{C}$ is the unique positive semidefinite square root of $C$ and $\sqrt{C_{P'{\bm{I}}'}}^{-1}$ is the Moore-Penrose pseudoinvese defined on the support of $\sqrt{C_{P'{\bm{I}}'}}$. 
By using the relation $C\otimes\id\dket{\id}=\id\otimes C^T\dket{\id}$ and the cyclic property of the trace, we verify that $S= \tr_\text{aux} \Big((E_{P{\bm{I}}\text{aux}}^{T_\text{aux}}\otimes \id_{\bm{O}F})\, (\id_{P{\bm{I}}}\otimes D_{\text{aux}{\bm{O}}F})\Big)$.

The definition of parallel superchannels prove convenient to investigate the action of parallel superchannels on an arbitrary input channel $\map{C}:\H_{\bm{I}}\to\H_{\bm{O}}$. In particular, let $C_{\bm{IO}}\in\L(\H_{\bm{I}}\otimes\H_{\bm{O}})$ be the Choi operator of $\map{C}$, Refs.~\cite{chiribella07,chiribella09networks,ebler16,quintino21unitary} show that
\begin{align} \label{eq:link}
    \tr_{\bm{IO}}\Big(S \; \left(\id_P\otimes C_{\bm{IO}}^T \otimes \id_F\right) \Big) \in\L(\H_P\otimes\H_F)
    %\map{D} \circ \Big(\map{C}\otimes \map{\id_\text{aux}}\Big)\circ \map{E}.
      %\tr_{\bm{IO}\text{aux}}\left(\Big((E_{P\bm{I}\text{aux}}^{T_\text{aux}}\otimes \id_{\bm{O}F})\, \id_{PI}\otimes D_{\text{aux}\bm{O}F})\Big) \; \id_P\otimes C_{\bm{IO}} \otimes \id_F \right) \\
\end{align}
is the Choi operator of 
\begin{align} \label{eq:link2}
    \map{D} \circ \Big(\map{C}\otimes \map{\id_\text{aux}}\Big)\circ \map{E} : \L(\H_P)\to\L(\H_F).
\end{align}

\subsection{The performance operator}
\label{subsec:performance_operator}
A useful tool to evaluate the performance of superchannels for transforming $k$ uses of a unitary operation $U_d$ into $f(U_d)$ is given by the \textit{performance operator}\footnote{Here we use the definition of Ref.~\cite{quintino21unitary}, which is differs from the definition of Ref.~\cite{ebler16} by a global transposition. We emphasise that for solve opmisation problems involving superchannels, the definitions from \cite{ebler16} and \cite{quintino21unitary} are fully equivalent since $S$ is a superchannel if and only if $S^T$ is a superchannel. }
~\cite{ebler16,quintino21unitary}, defined as
%\small
\begin{align}\label{eq:OMEGA}
		\Omega:=	\frac{1}{d^2}\int_{U_d\in\SU(d)}  \dketbra{{f(U_d)}}{{f(U_d)}}_{PF}\otimes\dketbra{U_d^{\otimes k}}{U_d^{\otimes k}}_{\bm{IO}}^T \dif U_d.
\end{align}
%\normalsize
From Eq.~\eqref{eq:link} and Eq.~\eqref{eq:link2} we see that using a superchannel $S$ with encoder channel $\map{E}$ and decoder channel $\map{D}$, the average fidelity can be written as
\small
\begin{align}\label{eq:OMEGAisGOOD}
\tr(S\Omega)=& \frac{1}{d^2}\int \tr \left(S\; \dketbra{{f(U_d)}}{{f(U_d)}}_{PF}\otimes\dketbra{U_d^{\otimes k}}{U_d^{\otimes k}}_{\bm{IO}}^T\right)\dif U_d  \\
=& \int \frac{1}{d^2} \tr_{PF}\Big(\dketbra{{f(U_d)}}{{f(U_d)}}_{PF} \\ 
\nonumber  & \left[\tr_{\bm{IO}}\left(S\; \id_P\otimes\dketbra{U_d^{\otimes k}}{U_d^{\otimes k}}_{\bm{IO}}^T \otimes \id_F\right)\right]\Big)\dif U_d  \\
	=&\int F\Big(\map{D}\circ \map{U}_d^{\otimes k} \circ \map{E} , \, \map{f({U_d})}\Big)\dif U_d  .
\end{align}
\normalsize
Hence, when the performance operator $\Omega$ is given, by combining Eq.~\eqref{eq:OMEGAisGOOD} with the superchannel characterisation given by Eq.~\eqref{eq:superchannel_characterised}, the problem of finding the optimal average fidelity with a parallel superchannel reads
\begin{align}
	&\max{ \tr(S \Omega)}\label{prob:primal1} \\
	\text{s.t. :} \quad & 
	S\geq0 \label{prob:primal2}\\
	&\tr_F(S)= C_{P\bm{I}} \otimes \id_{\bm{O}} \label{prob:primal3}\\
	&\tr_{\bm{I}}(C) = \id_{P}. \label{prob:primal4}
%	 S= \tr_\text{aux} \Big(E_{PI\text{aux}}\otimes \id_{\bm{O}F})\, (\id_{PI}\otimes D_{P\bm{I}\text{aux}})\Big) \\
%	& E\geq0, \quad \tr_{\text{aux}I}(E) = \id_{P} \\
%    & D\geq0, \quad \tr_{F}(D) = \id_{\text{aux}O}.
\end{align}
Reference~\cite{ebler16,quintino21unitary} shows that the dual optimisation problem is given by
\begin{align}
	&\min{ c} \label{prob:dual1}\\
	\text{s.t. :} \quad & 
	\Omega \leq  c \hat{S} \label{prob:dua2}\\
	& \hat{S}=W_{P\bm{IO}}\otimes \id_F \label{prob:dual3}\\
    & \tr_{\bm{O}}(W)=\rho_P \otimes \id_{\bm{I}} \label{prob:dual4}\\
    & \tr(\rho)=1. \label{prob:dual5}
\end{align}
As the notation suggests, $\rho \in \L(\H_P)$ represents a quantum state, $W\in\L(\H_P\otimes\H_{\bm{I}}\otimes\H_{\bm{O}})$ is a causally ordered process~\cite{chiribella09networks,oreshkov11}, and $\hat{S}$ lies in the dual affine set of (parallel) superchannels~\cite{ebler16,bavaresco21}. By definition of the dual problem, any feasible point for the dual problem constitutes as an upper bound for the primal. Additionally, for this particular problem we have strong duality, which means that the optimal value for the dual problem coincides with the optimal value for the primal problem. 

Strong duality can be shown by making use of the Slater conditions, which state that, if there exists an operator which satisfy the constraints of the primal (or dual) problem strict, the optimisation problem has strong duality. For our particular problem, it is then enough to find an operator $S\in\L(\H_P\otimes\H_{\bm{I}}\otimes\H_{\bm{O}}\otimes\H_F)$ which is strictly positive $S>0$, $\tr_F(S)= C_{P\bm{I}} \otimes \id_{\bm{O}}$, and $\tr_{\bm{I}}(C) = \id_{P}$. One such choice which trivially satisfy such strictly properties is $S=\id_P\otimes\frac{\id_{\bm{O}}}{d_{\bm{I}}}\otimes\id_{\bm{O}}\otimes\frac{\id_F}{d_F}$.

Since $\dketbra{U_d}{U_d}^T=\dketbra{\overline{U}_d}{\overline{U}_d}$, the performance operator for transforming $k$ calls of a unitary $U_d$ into its complex conjugate $U_d$ reads as
\small
\begin{align}\label{eq:OMEGAconj}
\Omega=&\frac{1}{d^2}\int_{U_d\in\SU(d)}  \dketbra{\overline{U}_d}{\overline{U}_d}_{PF}\otimes\dketbra{U_d^{\otimes k}}{U_d^{\otimes k}}_{\bm{IO}}^T \dif U_d  ,\\
=&\frac{1}{d^2}\int_{U_d\in\SU(d)}  \dketbra{\overline{U}_d}{\overline{U}_d}_{PF}\otimes\dketbra{\overline{U}_d^{\otimes k}}{\overline{U}_d^{\otimes k}}_{\bm{IO}} \dif U_d  ,\\
=&\frac{1}{d^2}\int_{U_d\in\SU(d)}  \dketbra{{U_d}}{{U_d}}_{PF}\otimes\dketbra{{U_d^{\otimes k}}}{{U_d^{\otimes k}}}_{\bm{IO}} \dif U_d  \\
=&\frac{1}{d^2}\int 
\Big(\id_P\otimes \id_{\bm{I}}^{\otimes k}\otimes {(U_d^{\otimes k})_{\bm{O}}}\otimes{(U_d)}_F\Big) \\ \nonumber
& \hspace{0.1mm}
\dketbra{{\id}}{{\id}}_{PF}\otimes\dketbra{{\id}}{{\id}}_{\bm{IO}} 
\Big(\id_P\otimes \id_{\bm{I}}^{\otimes k}\otimes {(U_d^{\otimes k})_{\bm{O}}}\otimes {(U_d)}_F\Big)^\dagger \dif U_d  .
\end{align}
\normalsize
Hence, the performance operator respects the commutation relation
\begin{align}
    [\Omega, \id_P\otimes \id_{\bm{I}}^{\otimes k}\otimes {(U_d^{\otimes k})_{\bm{O}}}\otimes {(U_d)}_F]=0, \quad \forall U_d\in\SU(d),
\end{align}
and since $\id\otimes U_d\ketbra{\id}{\id}=U_d^T\otimes \id\ketbra{\id}{\id}$, we also have 
\begin{align} \label{eq:commute2}
    [\Omega, {(U_d)}_P\otimes{(U_d^{\otimes k})}_{\bm{I}}\otimes {\id}_{\bm{O}}^{\otimes k}\otimes \id_F]=0, \quad \forall U_d\in\SU(d).
\end{align}
Now, let $\{B^i\}_i$ be an orthonormal basis for the subspace spanned by operators commuting with $U_d^{\otimes (k+1)}$, \textit{i.e.,} $\tr(B^i{B^j}^\dagger)=d_i\delta_{ij}$, and  $d_i:=\tr(B^i {B^i}^{\dagger})$. Eq.~(\ref{eq:commute2}) ensures that we can write the performance operator as
\begin{align}
    \Omega=\sum_i B^i_{P{\bm{I}}}\otimes {\Omega^i_{\bm{O}F}}
\end{align}
for some set of operators $\Omega^i_{\bm{O}F}\in \L(\H_{\bm{O}}\otimes\H_F)$. As noticed in Ref.~\cite{quintino21unitary}, we can find such operators by using the orthogonality relation, $\tr(B^i{B^j}^\dagger)=d_i\delta_{ij}$. More explicitly, by making use of $\tr_{2}\dketbra{\id}{\id}=\id_1$ we can write
\footnotesize
\begin{align}
    d_i\Omega^i_{\bm{O}F}=&\tr_{P\bm{I}} \Big[\big( {(B_{P\bm{I}}^{i})^{\dagger}} \otimes \id_{\bm{O}F}\big) \Omega \Big] \\
=&\frac{1}{d^2}\int  \tr_{P\bm{I}} %\Big(\big(\id_{P \bm{I}}\otimes B_{\bm{O}F}^i^\dagger \big) 
\Big[\big( (B_{P\bm{I}}^i)^\dagger \otimes \id_{\bm{O}F}\big) \\ 
\nonumber &
\Big(\id_P\otimes \id_{\bm{I}}^{\otimes k}\otimes {(U_d^{\otimes k}})_{\bm{O}}\otimes {U_d}_F\Big) \\ 
\nonumber &
\dketbra{{\id}}{{\id}}_{PF}\otimes\dketbra{{\id}}{{\id}}_{\bm{IO}} 
\Big(\id_P\otimes \id_{\bm{I}}^{\otimes k}\otimes {(U_d^{\otimes k}})_{\bm{O}}\otimes {U_d}_F\Big)^\dagger \dif U_d 
\Big] \\
=&\frac{1}{d^2}\int  \tr_{P\bm{I}} %\Big(\big(\id_{P \bm{I}}\otimes B_{\bm{O}F}^i^\dagger \big) 
\Big[
\Big(\id_P\otimes \id_{\bm{I}}^{\otimes k}\otimes {(U_d^{\otimes k}})_{\bm{O}}\otimes {U_d}_F\Big) \\ 
\nonumber &
\big( \id_{P\bm{I}}\otimes \overline{B}_{F\bm{O}}^i \big)
\dketbra{{\id}}{{\id}}_{PF}\otimes\dketbra{{\id}}{{\id}}_{\bm{IO}}\\ 
\nonumber &
\Big(\id_P\otimes \id_{\bm{I}}^{\otimes k}\otimes (U_d^{\otimes k})_{\bm{O}}\otimes {U_d}_F\Big)^\dagger \dif U_d 
\Big] \\
=&\frac{1}{d^2}\int %\Big(\big(\id_{P \bm{I}}\otimes B_{\bm{O}F}^i^\dagger \big) 
\Big((U_d^{\otimes k})_{\bm{O}}\otimes {U_d}_F\Big)
 \overline{B}_{F\bm{O}}^i 
\Big((U_d^{\otimes k})_{\bm{O}}\otimes {U_d}_F\Big)^\dagger \dif U_d  \\
=&\frac{1}{d^2}\int %\Big(\big(\id_{P \bm{I}}\otimes B_{\bm{O}F}^i^\dagger \big) 
\Big(\overline{(U_d^{\otimes k})_{\bm{O}}\otimes {U_d}_F}\Big)
\overline{B}_{F\bm{O}}^i 
\Big(\overline{(U_d^{\otimes k})_{\bm{O}}\otimes {U_d}_F}\Big)^\dagger \dif U_d  \\
=&\frac{1}{d^2}\int %\Big(\big(\id_{P \bm{I}}\otimes B_{\bm{O}F}^i^\dagger \big) 
\overline{B}_{F\bm{O}}^i 
\Big(\overline{(U_d^{\otimes k})_{\bm{O}}\otimes {U_d}_F}\Big)
\Big(\overline{(U_d^{\otimes k})_{\bm{O}}\otimes {U_d}_F}\Big)^\dagger \dif U_d  \\
=&\frac{1}{d^2}\int %\Big(\big(\id_{P \bm{I}}\otimes B_{\bm{O}F}^i^\dagger \big) 
\overline{B}_{F\bm{O}}^i \dif U_d  \\
=&\frac{1}{d^2} 
\overline{B}_{F\bm{O}}^i.   
\end{align}
This yields the following form of the performance operator for unitary complex conjugation
\normalsize
\begin{align}
\label{eq:Omega_basis}
  \Omega =  \frac{1}{d^2}\sum_i \frac{B^i_{P{\bm{I}}}\otimes \overline{B}^i_{F\bm{O}}}{d_i}.
\end{align}

\subsection{Evaluating the performance of the superchannel presented in Sec.~\ref{sec:construction}}
\label{sec:primal_solution}

In this section, we show that the superchannel presented in Sec.~\ref{sec:construction} attains optimal performance. The proof of its optimality is presented in subsection~\ref{subsec:upper_bound}.

    The superchannel $S\in\L(\H_P\otimes\H_{\bm{I}}\otimes\H_{\bm{O}}\otimes\H_F)$ described in Eq.~\eqref{eq:E} and Eq.~\eqref{eq:D} reads as
\begin{align} \label{eq:best_superchannel}
    S:=&E_{P\bm{I}} \otimes D_{\bm{O}F} \\
    E:=&\frac{d}{\binom{d}{k+1}} A(d,k+1) \\
    D:=&\frac{\binom{d}{k}}{\binom{d}{k+1}} A(d,k+1)+\Big(\id_d^{\otimes k}- A(d,k)\Big)\otimes \sigma, 
\end{align}
where $A(d,k)$ is the projector onto the antisymmetric space of $\mathbb{C}_d^{\otimes k}$ and $\sigma$ is an arbitrary quantum state.

\begin{lemma} \label{Lemma:1}
The quantum circuit described in Sec.~\ref{sec:construction} transforms $k$ uses of a $d$-dimensional unitary operation $U_d$ into its complex conjugation $\overline{U_d}$ with an average fidelity respecting ${\mean{F}\geq\frac{k+1}{d(d-k)}}$.
\end{lemma}
\begin{proof}
We prove the theorem by showing that
\begin{equation}
    \tr(\Omega\,S)=\frac{k+1}{d(d-k)}
\end{equation}
with $\Omega$ being the performance operator for the unitary complex conjugation transformation
\begin{align}
\Omega=	\int_{U_d\in\SU(d)}  \dketbra{U_d}{U_d}_{PF}\otimes\dketbra{U_d^{\otimes k}}{U_d^{\otimes k}}_{\bm{IO}} \dif U_d  ,
\end{align}
and $S$ the superchannel described in Eq.~\eqref{eq:best_superchannel}.

As shown in Eq.~\eqref{eq:Omega_basis} and Ref.~\cite{quintino21unitary}, the performance operator for unitary complex conjugation may be written as
\begin{align}
  \Omega =  \sum_i \frac{1}{d^2}\frac{B^i_{P\bm{I}}\otimes \overline{B^i_{F\bm{O}}}}{d_i}
\end{align}
where $\{B^i\}_i$ is an orthonormal basis, \textit{i.e.,} $\tr(B^i {B^j}^{\dagger})=\delta_{ij} d_i$, for the subspace spanned by operators commuting with all unitary operators $U^{\otimes (k+1)}$. Since the projector onto the antisymmetric space $A(d,k+1)$ respects $A(d,k+1) U^{\otimes (k+1)} = U^{\otimes (k+1)} A(d,k+1) $, we can set $A(d,k+1)$ as $B^1$, one of the operators of the orthonormal basis $\{B^i\}_i$. We can then write 
%\small
\begin{align}
  \Omega =  \frac{1}{d^2}\frac{A(d,k+1)_{P\bm{I}}\otimes A(d,k+1)_{F\bm{O}}}{\tr\big(A(d,k+1)\big)} + \frac{1}{d^2}\sum_{i>1} \frac{B^i_{P{\bm{I}}}\otimes \overline{B^i_{{F\bm{O}}}}}{d_i}.
\end{align}
%\normalsize 

Now, since $E_{P\bm{I}}=\frac{d}{\binom{d}{k+1}} A(d,k+1)$ and ${\tr(A(d,k+1)\,B^i)=0} $ for $i>1$, it holds that
\begin{align}
\tr(\Omega \, S) = \frac{1}{d^2} \frac{\tr\Big(E\,A(d,k+1)\Big)\tr\Big(D\,A(d,k+1)\Big)}{\tr\big(A(d,k+1)\big)}.
\end{align}
We now recall from Sec.~\ref{sec:construction} that $\tr\Big(A(d,k)\Big)=\binom{d}{k}$, and since $A(d,k)^2=A(d,k)$, we have that 
\small
\begin{align}
&\tr(\Omega \, S) = \frac{1}{d^2}
\frac{d\binom{d}{k} \binom{d}{k+1}\binom{d}{k+1}}{\binom{d}{k+1}\binom{d}{k+1}\binom{d}{k+1}}
\\ 
\nonumber &+ \frac{1}{d^2} \frac{d\binom{d}{k+1}}{\binom{d}{k+1}\binom{d}{k+1}}
\tr\Big(A(d,k+1)\, \big(\id_d^{\otimes k}-A(d,k)\big)\otimes \sigma \Big).
\end{align}
\normalsize
Since $\id_d^{\otimes k}-A(d,k)$ is positive semidefinite, we have that 
\small
\begin{align}
\tr(\Omega \, S) \geq& \frac{1}{d^2}
\frac{d\binom{d}{k} \binom{d}{k+1}\binom{d}{k+1}}{\binom{d}{k+1}\binom{d}{k+1}\binom{d}{k+1}} \\
=&\frac{\binom{d}{k}} {d\binom{d}{k+1}} \\
=&\frac{d!}{k!(d-k)(d-k-1)!} \frac{(k+1)k!(d-k-1)!}{d\,d!} \\
=&\frac{k+1}{d(d-k)}.
\end{align}
\normalsize
%\begin{align}
%\tr(\Omega \, S) \geq& \frac{1}{d^2}
%\frac{d\binom{d}{k} \binom{d}{k+1}\binom{d}{k+1}}{\binom{d}{k+1}\binom{d}{k+1}\binom{d}{k+1}}
%=&\frac{\binom{d}{k}} {d\binom{d}{k+1}} \\
%=&\frac{d!}{k!(d-k)(d-k-1)!} \frac{(k+1)k!(d-k-1)!}{d\,d!} \\
%=\frac{k+1}{d(d-k)}.
%\end{align}
This finishes the proof.
\end{proof}
As we are going to show next, this protocol cannot be improved. Its average fidelity is exactly $\mean{F}=\frac{k+1}{d(d-k)}$ and we have that $\tr\Big(A(d,k+1)\, \big(\id_d^{\otimes k}-A(d,k)\big)\otimes \sigma \Big)=0$ for any quantum state $\sigma$.

\subsection{A tight upper bound for the unitary conjugation problem} \label{subsec:upper_bound}
To prove that the superchannel $S$ presented in section~\ref{sec:construction} is indeed optimal, we provide here a feasible solution for the dual problem defined by equations~\eqref{prob:dual1}-\eqref{prob:dual5}, and show that its solution matches with the solution of the respective primal problem from~\eqref{prob:primal1}-\eqref{prob:primal4} solved in section~\ref{sec:primal_solution}. In particular, we show that for our proposed ansatz, the value of the constant $c$ from~\eqref{prob:dual1} matches with the value from~\eqref{prob:primal1} -- reproducing the value of optimal fidelity from Theorem~\ref{thm:main}. In order to proceed, we need to introduce basic knowledge from representation theory, including the famous Schur-Weyl duality~\cite{FultonSchur,harrisBook}, Young diagrams and their relation to irreducible representation (irrep) of the symmetric group~\cite{FultonSchur,harrisBook}. We give a brief introduction in the following subsection, while the solution of the dual problem for the unitary conjugation problem is derived later.

\subsubsection{Young diagrams, Schur-Weyl duality and irreducible representations of the symmetric group}
\label{sec:irreps}
A partition $\lambda$ of a natural number $n$, which we denote as $\lambda \vdash n$, is a sequence of positive numbers $\lambda=[\lambda_1,\lambda_2,\ldots,\lambda_m]$, such that
\begin{align}
\label{eq:yng}
\lambda_1\geq \lambda_2\geq \cdots \geq \lambda_m\geq 0\qquad \sum_{i=1}^m\lambda_i=n.
\end{align}
Every partition can be visualised as a \textit{Young diagram} - a collection of boxes arranged in left-justified rows. With every Young diagram we associate it height $H(\lambda)$, which is defined as a number of boxes in the first (the longest one) column. See Figure~\ref{Fig2} for illustration of the above definitions.
\begin{figure}
\begin{center}
	\includegraphics[width=0.49\textwidth]{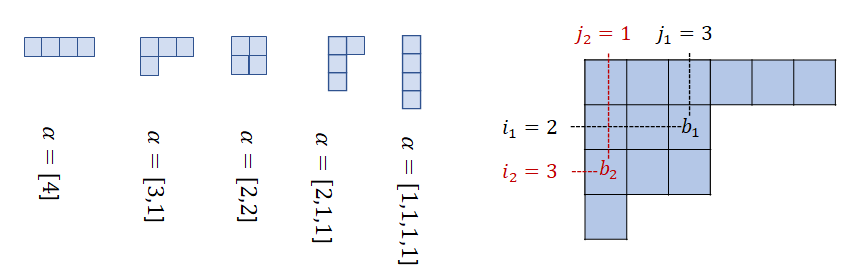}
\end{center}
\caption{The left figure presents five possible Young diagrams for $n=4$. These Young diagrams corresponds to all possible abstract irreducible representations of $S(4)$. Considering representation space $(\mathbb{C}^d)^{\otimes 4}$, there appear only irreps for which height  of the corresponding Young diagram is no larger than $d$. For example, considering qubits ($d=2$) we have only three Young diagrams: $(4),(3,1),(2,2)$. In particular no antisymmetric space exists. The right figure presents graphical example of assigning the coordinates to a box within a Young diagram. Here we have $\mu=[6,3,3,1]$ and the box $b_1$ has the coordinates $(i_1,j_1)=(2,3)$, while the box $b_2$ has the coordinates $(i_2,j_2)=(3,1)$. The axial distance between these to boxes, according to~\eqref{eq:ax_dist} is equal $d(b_1,b_2)=|i_1-i_2|+|j_1-j_2|=3$. } \label{Fig2}
\end{figure}

Now we introduce the notion of a step representation for a given Young diagram~\cite{vershik}. Let $k_1,k_2,\ldots,k_s$ be multiplicities of numbers $\lambda_i$ in the non-increasing sequence~\eqref{eq:yng}. Let us also define $p_i:=\lambda_{k_i}-\lambda_{k_{i+1}}$ for $i=1,2,\ldots,s$, where we assume $\lambda_{k_{s+1}}=0$.
Observe that 
\begin{equation}
    \sum_{1\leq i\leq j\leq s} k_i p_j = n,
\end{equation}
and a Young diagram is uniquely determined by the sequence, 
\begin{equation}
\label{e:seqks}
(k_1,p_1,k_2,p_2,\ldots,k_s,p_s).
\end{equation}
Let us name the sequence a \textit{step representation} of a diagram $\lambda$, see Figure~\ref{Fig2a}.
\begin{figure}[!h]
\begin{center}
		\includegraphics[width=0.49\textwidth]{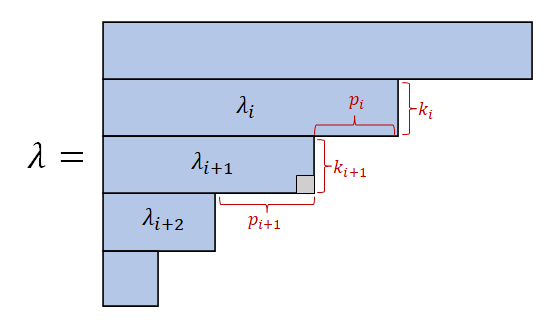}
\end{center}
\caption{The graphic presents parameters describing a step representation for an arbitrary Young diagram $\lambda$. The numbers $k_i,k_{i+1}$ count the number of appearances of row of lengths $\lambda_i,\lambda_{i+1}$ respectively, while the numbers $p_i,p_{i+1}$ are the respective differences $\lambda_{i+1}-\lambda_i, \lambda_{i+1}-\lambda_{i+2}$ between the rows lengths. By a grey square we depicted one corner box $c_{i+1}$ whose coordinates can be calculated through expression~\eqref{e:cj}.}
\label{Fig2a}
\end{figure}

For every two boxes $b_1=(i_1,j_1)$, $b_2=(i_2,j_2)$ let us define their \textit{axial distance} (see Figure~\ref{Fig2}) as:
\begin{equation}
\label{eq:ax_dist}
    d(b_1,b_2):=|i_1-i_2|+|j_1-j_2|.
\end{equation}
Let $c_1,c_2,\ldots,c_s$ be corner boxes
\begin{equation}
\label{e:cj}
    c_j=(k_1+k_2+\ldots+k_j,p_j+p_{j+1}+\ldots+p_s), \quad j=1,2,\ldots,s.
\end{equation}
For further applications let us observe that for $i<j$,
\begin{equation}
    d(c_i,c_j)=p_i+\ldots+k_j,
\end{equation}
where the above notation means the sum of all numbers occurring in the sequence \eqref{e:seqks} between $p_i$ and $k_j$ including them.

For a fixed number $n$, the number of Young diagrams determines the number of nonequivalent irreps of  $S_n$ in an abstract decomposition.   However, working in the representation space $\mathcal{H\equiv (\mathbb{C}}^{d})^{\otimes n}$, for every decomposition of $S_n$ into irreps we take Young diagrams $\lambda$ with height $H(\lambda)$ of a given $\lambda$ of at most $d$. For example, in the case of qubits ($d=2$), we allow only for diagrams up to two rows.

Suppose we have $\alpha \vdash n-1$ and $\lambda \vdash n$. Writing $\lambda \ni \alpha$ we consider such Young diagrams $\lambda$ which can be obtained from $\alpha$ by adding a single box (green colour). Similarly, writing $\alpha \in \lambda$ we consider such Young  diagrams $\alpha$, which can be obtained from $\lambda$ by removing a single box. The procedure of adding/removing a box from Young diagrams is illustrated in Figure~\ref{Fig3}.  
\begin{figure}[!h]
\begin{center}
	\includegraphics[width=0.49\textwidth]{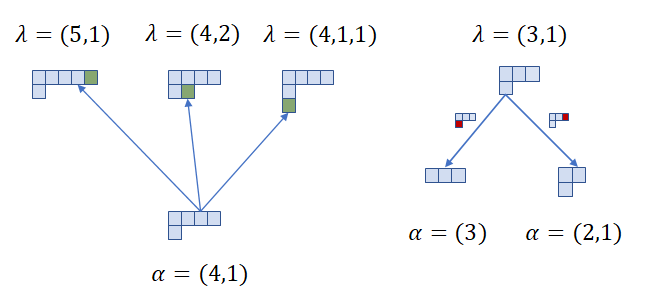}
\end{center}
\caption{The left graphic presents possible Young diagrams $\lambda \vdash 6$, which can be obtained from Young diagram $\alpha=[4,1]$ by adding a single box, depicted here by the green colour. In this particular case, by writing $\lambda \ni \alpha$, we take $\lambda$ represented only by these three diagrams. In the same manner we define subtracting of a box from a Young diagram. This is depicted on the right graphic where we remove a single box (red colour) from the diagram $\lambda=[3,1]$ and obtain two three-boxed Young diagrams $\alpha=[3]$ and $\alpha=[2,1]$. Diagrams $\alpha$ obtained from a diagram $\lambda$ by removing a single box we denote as $\alpha\in\lambda$.} \label{Fig3}
\end{figure}

Having the above definitions and notations we are in position to summarise more results on representation theory used further in this manuscript.
We start from the celebrated Schur-Weyl duality relating irreps of the general linear group $GL(d)$ and symmetric group $S_n$. 
Namely, it is known that the diagonal action of the general linear group $GL(d)$ of invertible complex matrices and of the symmetric group on $(\mathbb{C}^d)^{\otimes n}$ commute:
\begin{align}
V(\pi)(X\otimes \cdots \otimes X)=(X\otimes \cdots \otimes X)V(\pi),
\end{align}
where $\pi \in S_n$ and $X\in GL(d)$. Due to the above relation, there exists a basis called the Schur basis, producing a decomposition into irreps of $GL(d)$ and $S_n$ simultaneously. Hence, for the whole space $(\mathbb{C}^d)^{\otimes n}$ we have the following
\begin{theorem}[Schur-Weyl duality]
	\label{SW}
	The tensor product space $(\mathbb{C}^d)^{\otimes n}$ can be decomposed as
	\begin{align}
	\label{eq:SW}
	(\mathbb{C}^d)^{\otimes n}=\bigoplus_{\substack{\lambda \vdash n \\ H(\lambda)\leq d}} \mathcal{U}_{\lambda}\otimes \mathcal{S}_{\lambda},
	\end{align}
	where the symmetric group $S_n$ acts on the space $\mathcal{S}_{\lambda}$ and the general linear group $GL(d)$ acts on the space $\mathcal{U}_{\lambda}$, labelled by the same Young diagram $\lambda \vdash n$.
\end{theorem}

In fact, for our purposes we need the Schur-Weyl duality not for $GL(d)$ but rather its subgroup $SU(d)$, see~\cite{quintino21unitary}. However, it turns out that 
any irreducible representation of $GL(d)$ on $(\mathbb{C}^d)^{\otimes n}$ is an irreducible representation of $SU(d)$ on $(\mathbb{C}^d)^{\otimes n}$ and vice versa. In particular, it means that  that also the linear invariants of $SU(d)$ are
those of $GL(d)$, and they belong to the group algebra of $S_n$ denoted by $\mathbb{C}[S_n]$. This however means that we can write the same statement in Theorem~\ref{SW} for group $SU(d)$~\cite{Goodman}.
From the decomposition given in Theorem~\ref{SW} we deduce that for a given irrep $\lambda$ of $S_n$, the space $\mathcal{U}_{\lambda}$ is the multiplicity space of dimension $m_{\lambda}$ (multiplicity of irrep $\lambda$), while the space $\mathcal{S}_{\lambda}$ is the representation space of dimension $d_{\lambda}$ (dimension of irrep $\lambda$). Every operator commuting with the diagonal action of $U^{\otimes n}$, due to the Schur's lemma must be proportional to identity on  spaces $\mathcal{U}_{\lambda}$ and admits non-triviall parts on spaces $\mathcal{S}_{\lambda}$. Conversely, any operator commuting with action of the permutation group $S_n$ is non-trivially supported on irreducible spaces $\mathcal{U}_{\lambda}$. In the former case, i.e. in  very irreducible space $\mathcal{S}_{\lambda}$ one can construct an orthonormal irreducible basis $\{|\lambda,i\>\}$, where $i=1,\ldots,d_{\lambda}$, for example by exploiting Young-Yamanouchi construction~\cite{Ceccherini}. 
With this vectors we associate an irreducible basis operators $E^{\lambda}_{ij}$, for $i,j=1,\ldots,d_{\lambda}$ admitting the following form:
\begin{align}
\label{eq:basisSWbasis}
E_{ij}^{\lambda}:=\id^{\mathcal{U}_{\lambda}}\otimes |\lambda,i\>\<\lambda,j|_{\mathcal{S}_{\lambda}}.
\end{align}
Sine the basis $\{|\lambda,i\>\}_{i=1}^{d_{\lambda}}$ is orthogonal in both indices the above operators fulfil the following properties, which are
\begin{align}
\label{irreps_basis_prop}
E_{ij}^{\lambda}E_{kl}^{\lambda'}=\delta^{\lambda \lambda'}\delta_{jk}E^{\lambda}_{il},\qquad \tr E_{ij}^{\lambda}=\delta_{ij}m_{\lambda}.
\end{align}
From the operators $E_{ij}^{\lambda}$, by using the relation $\sum_i |\lambda,i\>\<\lambda,i|_{\mathcal{S}_{\lambda}}=\id^{\mathcal{S}}_{\lambda}$, we can construct so-called Young projectors $P^{\lambda}$~\cite{Ceccherini} on the irreducible components $\lambda$:
\begin{align}
\label{Yng_proj}
P^{\lambda}:=&\sum_{i=1}^{d_{\lambda}}E_{ii}^{\lambda}=\id^{\mathcal{U}}_{\lambda}\otimes \id^{\mathcal{S}}_{\lambda},\\ P^{\lambda}P^{\lambda'}=&\delta^{\lambda \lambda'}P^{\lambda},\\ 
\tr P^{\lambda}=&m_{\lambda}d_{\lambda}.
\end{align}

Let us take $\lambda \vdash n$ and $\alpha,\beta \vdash n-1$, such that $\alpha,\beta \in \lambda$. Consider an arbitrary unitary irrep $\phi^{\lambda}(\pi)$ for some $\pi \in S_{n}$. Thus it can be always unitarily transformed (reduced) to block form consisting of elements of the subgroup $S_{n-1}\subset S_n$ as
\begin{align}
\phi^{\lambda}(\pi)=\left(\phi^{\lambda}_{i_{\lambda}j_{\lambda}}(\pi)\right)=\left(\phi^{\alpha \beta}_{i_{\alpha}j_{\beta}}(\pi)\right),
\end{align}
where the indices $i_{\lambda},j_{\lambda}$ are the standard matrix indices running from 1 to $d_{\lambda}$. From the above we see the following connection between the indices, namely for fixed $i_{\lambda},j_{\lambda}$ one has  $i_{\lambda}=(\alpha \in \lambda, i_{\alpha}),j_{\lambda}=(\beta\in\lambda,j_{\beta})$, where the indices $i_{\alpha},j_{\beta}$ run from 1 to $d_{\alpha},d_{\beta}$ respectively. Please see Figure~\ref{Fig4} where the graphical illustration of the above considerations is presented or see~\cite{Stu2020a} for more details. 
 The diagonal blocks of the matrix 
% The matrices on the diagonal of reduced matrix 
$\phi^{\lambda}(\pi)$, \textit{i.e}
the blocks with $\alpha=\beta$,  are of dimension of the corresponding irrep $\varphi^{\alpha}$ of $S_{n-1}$. It means that the both indices $i_{\alpha},j_{\alpha}$ run from 1 to $d_{\alpha}$. In particular, when one considers an irrep $\phi^{\lambda}(\pi')$ for $\pi'\in S_{n-1}$, we have
\begin{align}
\phi^{\lambda}(\pi')=\left(\delta^{\alpha \beta}\varphi^{\alpha}_{i_{\alpha}j_{\alpha}}(\pi')\right)=\bigoplus_{\alpha \in \lambda}\varphi^{\alpha}(\pi').
\end{align}

A similar procedure (induction) can be also done in the reverse order, i.e. starting from an irrep $\varphi^{\alpha}$ of $S_{n-1}$ we can induce irreps $\phi^{\lambda}$ of a larger group $S_n$.
\begin{figure}[!h]
\begin{center}
	\includegraphics[width=0.49\textwidth]{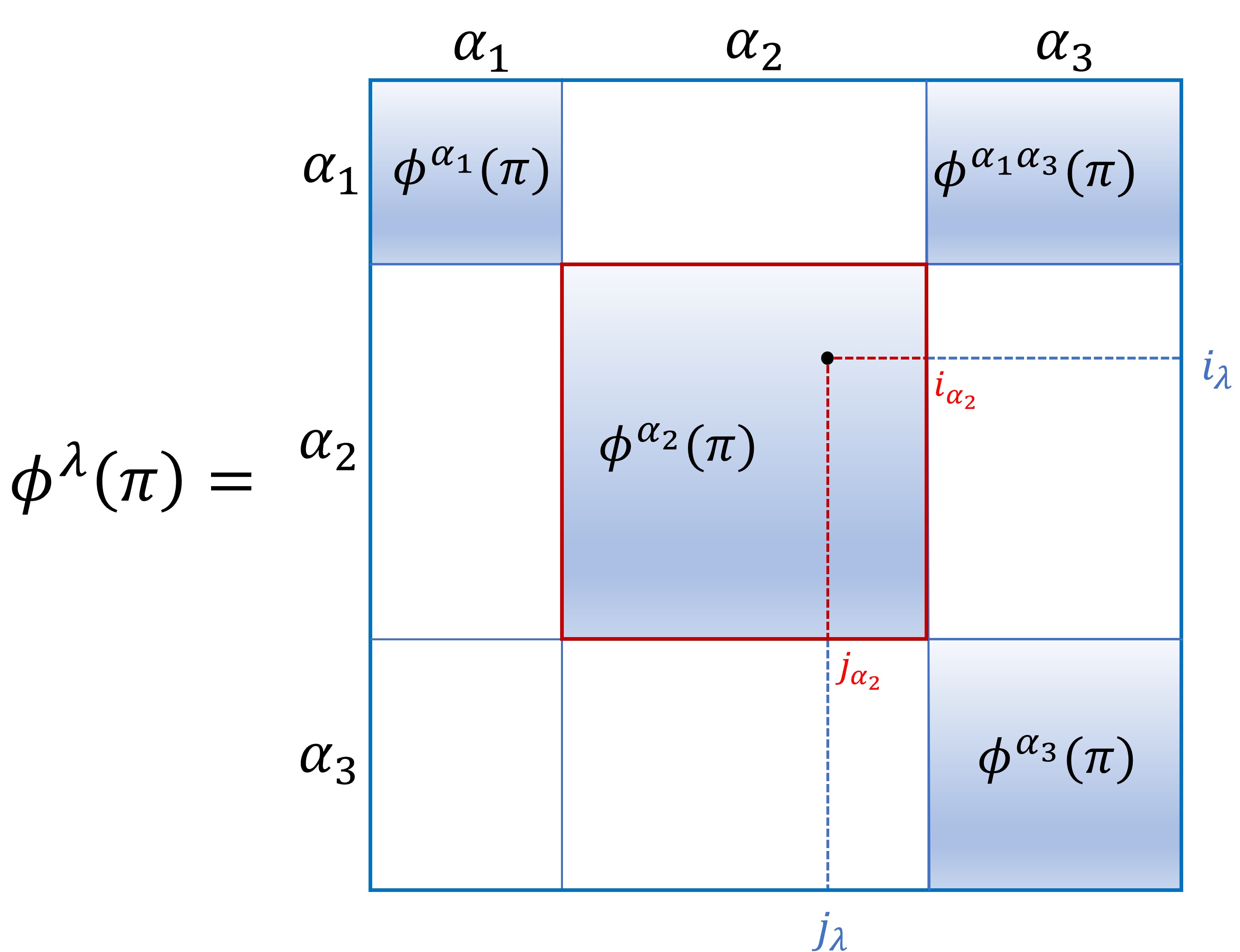}
	\includegraphics[width=0.49\textwidth]{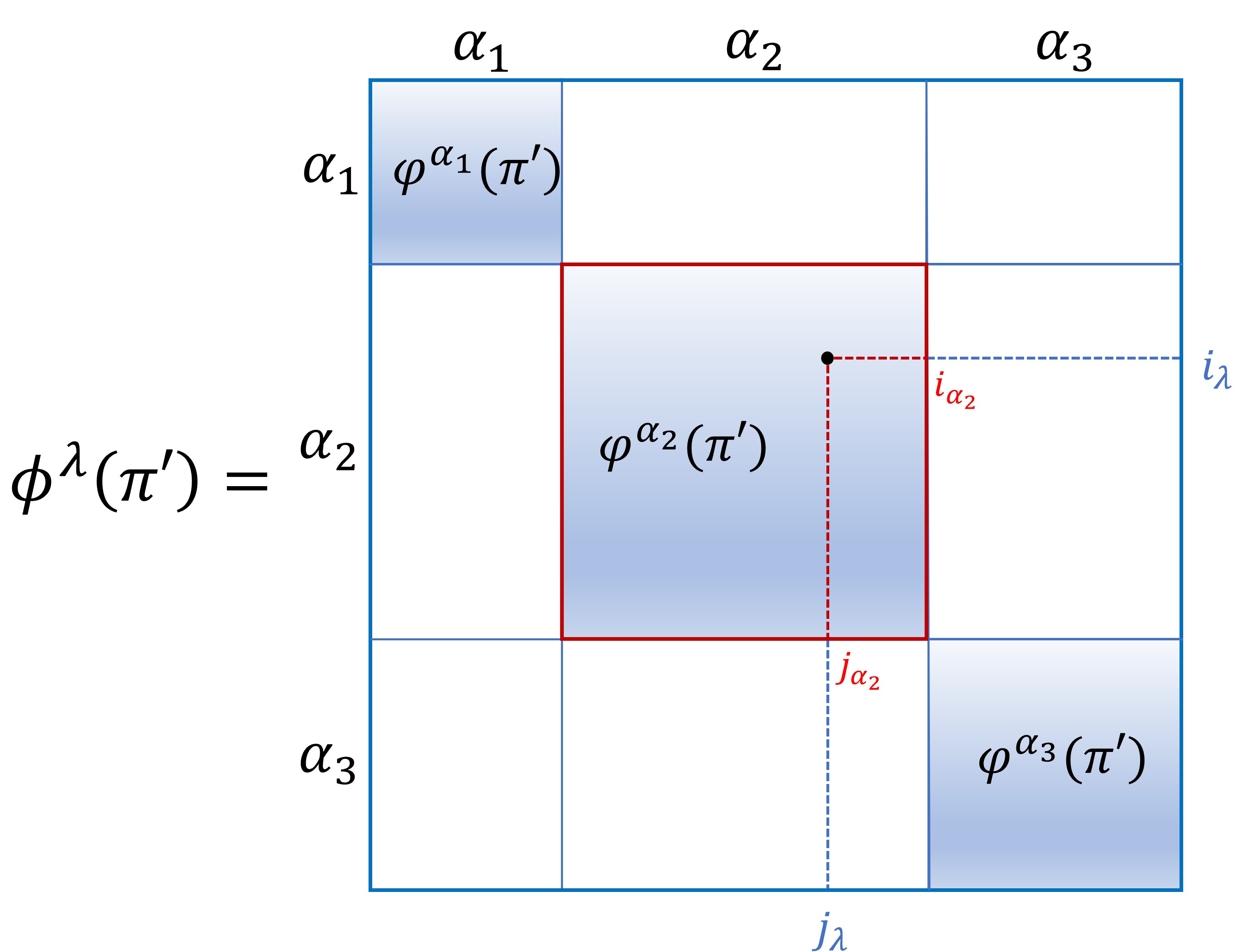}
\end{center}
\caption{The left-hand side graphic represents unitary transformed irrep $\phi^{\lambda}(\pi)=\left(\phi^{\lambda}_{i_{\lambda}j_{\lambda}}(\pi)\right)$ for some $\pi \in S_n$ and $\lambda \vdash n$ into a block form related to subgroup $S_{n-1}$. Each block $\left(\phi^{\alpha_i \alpha_j}(\pi)\right)$ is indexed by a pair of indices from $\{\alpha_1,\alpha_2,\alpha_3\}$ which are obtained from $\lambda$ by removing a single box, so by irreps of $S_{n-1}$. If we take $\pi'\in S_{n-1}$ then $\phi^{\lambda}(\pi')$ takes a block diagonal form, i.e. it can be represented as $\varphi^{\alpha_1}(\pi')\oplus \varphi^{\alpha_2}(\pi')\oplus \varphi^{\alpha_3}(\pi')$ - the right-hand side graphic. Moreover, the chain $\ldots S_{n-1}\subset S_{n} \subset \ldots$ is multiplicity free, which means that every $\alpha \vdash (n-1)$, for $\alpha\in\lambda$, occurs exactly once in an irrep $\lambda \vdash n$. There is also a one-to-one correspondence between matrix indices of $\phi^{\lambda}(\pi)$ and every sub-block $\left(\phi^{\alpha_i \alpha_j}(\pi)\right)$. Indeed, every element (black dot) can be indexed by a pair $(i_{\lambda},j_{\lambda})$, where $1\leq i_{\lambda},j_{\lambda}\leq d_{\lambda}$ or by a pair of indices $(i_{\alpha_2},j_{\alpha_2})$ from a sub-block $\phi^{\alpha_2}(\pi)$, where $1\leq i_{\alpha_2},j_{\alpha_2}\leq d_{\alpha_2}$, and we suppressed the double index. The same reasoning works also for off-diagonal blocks. Since the construction is multiplicity free, we have the following relations between dimensions of the irreps: $d_{\lambda}=d_{\alpha_1}+d_{\alpha_2}+d_{\alpha_3}$.} \label{Fig4}
\end{figure}
%Informally, this means that to a given $\alpha$ by adding one box in different ways allows for different induced representations $\mu$.
Furthermore, utilizing the relabelling $i_{\lambda}=(\alpha \in \lambda, i_{\alpha}),j_{\lambda}=(\beta\in\lambda,j_{\beta})$ introduced above leads us to the following notation for the operators $E_{ij}^\lambda$ from~\eqref{eq:basisSWbasis}:
\begin{notation}
\label{cor:lin_comb_of_E2}
    Every operator $E_{ij}^\lambda$, where $E_{ij}^\lambda$ are basis operators of the algebra $\mathbb{C}[S_{n}]$, can be re-written as%in terms of irreducible basis elements of the algebra $\mathbb{C}[S_{n-1}]$ as
    \begin{align}\label{eqn:collary3}
        E_{i_\lambda j_\lambda}^\lambda :=  E^{\alpha\alpha'}_{i_\alpha j_{\alpha'}}(\lambda),
    \end{align}
    where $\alpha,\alpha'\in\lambda$.
\end{notation}

The above considerations give us a way to represents elements of the algebra $\mathbb{C}[S_{n-1}]$ in terms of elements from the algebra $\mathbb{C}[S_{n}]$. Concretely, one can formulate the following:
\begin{lemma}
\label{cor:lin_comb_of_E}
    Operator $E_{ij}^\alpha \otimes \id$, where $E_{ij}^\alpha$ are basis operators of the algebra $\mathbb{C}[S_{n-1}]$, can be written in terms of basis elements of the algebra $\mathbb{C}[S_{n}]$ as
    \begin{align}\label{eqn:collary2}
        E_{i_\alpha j_\alpha}^\alpha \otimes \id = \sum_{\lambda\ni\alpha} E^{\alpha\alpha}_{i_\alpha j_\alpha}(\lambda).
    \end{align}
    Notice that the resulting irreducible basis operator $E^{\alpha\alpha}_{i_\alpha j_\alpha}(\lambda)$ from $\mathbb{C}[S_n]$ is block-diagonal in $\alpha\in\lambda$.
\end{lemma}

\begin{proof}
The operator $E_{i_\alpha j_\alpha}^\alpha \otimes \id$ clearly belongs to the algebra $\mathbb{C}[S_n]$, so in general it can be decomposed in terms of operators $E_{i_\lambda j_\lambda}^\lambda :=  E^{\beta\beta'}_{i_\beta j_{\beta'}}(\lambda)$ as
\begin{equation}
\label{eq:L2eq1}
E_{i_\alpha j_\alpha}^\alpha \otimes \id=\sum_{\lambda}\sum_{\beta,\beta'\in \lambda}\sum_{k_{\beta},l_{\beta'}}x_{k_{\beta}l_{\beta'}}^{\beta \beta'}(\lambda)E^{\beta \beta'}_{k_{\beta}l_{\beta'}}(\lambda),
\end{equation}
where the numbers $x_{k_{\beta}l_{\beta'}}^{\beta \beta'}(\lambda)$ are unknown coefficients. To determine the coefficients, let us compute overlap of the left-hand side of~\eqref{eq:L2eq1} with some basis operator $\left(E^{\gamma \gamma'}_{k_{\gamma}l_{\gamma'}}(\lambda)\right)^{\dagger}=E^{\gamma' \gamma}_{l_{\gamma'}k_{\gamma}}(\lambda)$:
\small
\begin{align}
\label{eq:L2eq2}
\tr\left[\left(E_{i_\alpha j_\alpha}^\alpha \otimes \id\right) E^{\gamma' \gamma}_{l_{\gamma'}k_{\gamma}}(\lambda)\right]&=\tr\left[E_{i_\alpha j_\alpha}^\alpha \tr_n \left(E^{\gamma' \gamma}_{l_{\gamma'}k_{\gamma}}(\lambda)\right)\right]\\&=
\frac{m_{\lambda}}{m_{\gamma}}\tr\left(E_{i_\alpha j_\alpha}^\alpha E_{l_{\gamma}k_{\gamma}}^{\gamma}\right)\delta_{\gamma \gamma'} \\
&=\frac{m_{\lambda}m_{\alpha}}{m_{\gamma}}\delta_{\gamma \gamma'}\delta_{\alpha \gamma}\delta_{j_{\alpha} l_{\gamma}}\delta_{i_{\alpha}k_{\gamma}}.
\end{align}
\normalsize
To get the third equality we use Lemma 7 from~\cite{Stu2020a}, to get the fourth equality we directly apply relations from~\eqref{irreps_basis_prop}. Now, let us apply the same procedure for the right-hand side of~\eqref{eq:L2eq1} but for some $\lambda'\neq \lambda$. In the same lines we get the following:
\small
\begin{align}
\label{eq:L2eq3}
&\sum_{\lambda}\sum_{\beta,\beta'\in \lambda}\sum_{k_{\beta},l_{\beta'}}x_{k_{\beta}l_{\beta'}}^{\beta \beta'}(\lambda)\tr\left(E^{\beta \beta'}_{k_{\beta}l_{\beta'}}(\lambda)E^{\gamma' \gamma}_{l_{\gamma'}k_{\gamma}}(\lambda')\right) \\
&=\sum_{\lambda}\sum_{\beta,\beta'\in \lambda}\sum_{k_{\beta},l_{\beta'}}m_{\lambda}x_{k_{\beta}l_{\beta'}}^{\beta \beta'}(\lambda)\delta_{\lambda\lambda'}\delta_{\beta'\gamma'}\delta_{l_{\beta'}l_{\gamma'}}\delta_{\beta\gamma}\delta_{k_{\beta}k_{\gamma}}\\
&=m_{\lambda}x_{k_{\gamma}l_{\gamma'}}^{\gamma \gamma'}(\lambda).
\end{align}
\normalsize
The final expressions in~\eqref{eq:L2eq2} and~\eqref{eq:L2eq3} must be equal giving us the following condition:
\begin{align}
\label{eq:L2eq4}
 x_{k_{\gamma}l_{\gamma'}}^{\gamma \gamma'}(\lambda)=\frac{m_{\alpha}}{m_{\gamma}}\delta_{\gamma \gamma'}\delta_{\alpha \gamma}\delta_{j_{\alpha} l_{\gamma'}}\delta_{i_{\alpha}k_{\gamma}}.
\end{align}
Inserting the right hand side of~\eqref{eq:L2eq4} into decomposition~\eqref{eq:L2eq1} we have:
\begin{align}
E_{i_\alpha j_\alpha}^\alpha \otimes \id&=\sum_{\lambda}\sum_{\gamma,\gamma'\in \lambda}\sum_{k_{\gamma},l_{\gamma'}}\frac{m_{\alpha}}{m_{\gamma}}\delta_{\gamma \gamma'}\delta_{\alpha \gamma}\delta_{j_{\alpha} l_{\gamma'}}\delta_{i_{\alpha}k_{\gamma}}E^{\gamma \gamma'}_{k_{\gamma}l_{\gamma'}}(\lambda)\\
&=\sum_{\lambda\ni\alpha} E^{\alpha\alpha}_{i_\alpha j_\alpha}(\lambda).
\end{align}
The final sum runs only over $\lambda$ which can be obtained from $\alpha$ by adding a single box. This finishes the proof.
\end{proof}

\subsubsection{Direct derivations of the upper bound for the unitary conjugation problem}
Due to the presentation of the performance operator $\Omega$ in~\eqref{eq:OMEGAconj}, we see that it commutes with the diagonal action of the unitary group. This means that its nontrivial part, due to the Schur-Weyl duality in Theorem~\ref{SW}, is supported on the representation space $\mathcal{S}_{\lambda}$ connected to the symmetric group. Therefore, the operators from~\eqref{eq:Omega_basis} can be chosen to be operators $E_{ij}^{\lambda}$ defined in~\eqref{eq:basisSWbasis}. Then the operator $\Omega$ from~\eqref{eq:Omega_basis} has the following presentation:
\begin{align}
\label{Omega_27}
\Omega=\frac{1}{d^2}\sum_{\substack{\lambda\vdash (k+1)\\H(\lambda)\leq d}}\sum_{p,q=1}^{d_{\lambda}}\frac{ (E_{pq}^{\lambda})_{\bm{I}P}\otimes (E_{pq}^{\lambda})_{\bm{O}F}}{m_{\lambda}}.
\end{align}
Since the operator $\Omega$ is non-trivial only on the space $\mathcal{S}_{\lambda}$, we assume the same property for the operator $\hat{S}$ from~\eqref{prob:dua2} given in formulation of the dual problem described in section~\ref{subsec:performance_operator}. For the time being, except additional symmetries imposed on the operator $\hat{S}$ we see that it can have an arbitrary form. However, the considered operator $\hat{S}$ must also satisfy constraints from~\eqref{prob:dua2}-\eqref{prob:dual5}, which obviously limits the compatible choices for $\hat{S}$. Since our goal is to find a tight upper bound which matches with the solution of the primal problem from section~\ref{sec:primal_solution}, it is sufficient to present a feasible solution satisfying aforementioned conditions. To this purpose, we choose the operator $\hat{S}$ to be of the following form: 
\begin{align}\label{eqn:ansatzW}
    \hat{S}=  \sum_{\substack{\alpha \vdash k\\H(\alpha)\leq d}} \sum_{k,l=1}^{d_{\alpha}}\frac{d}{m_\alpha} ({E_{kl}^\alpha})_{\bm{I}}\otimes \tau_P \otimes ({E_{kl}^\alpha})_{\bm{O}}  \otimes \tau_F.
\end{align}
In the above, the operators $\tau_P=\tau_F=\frac{\id}{d}$ are maximally mixed states on the respective systems.  It is easy to see that the operator $\hat{S}$ satisfies the constraints~\eqref{prob:dual3}-\eqref{prob:dual5} by construction. Indeed, we have the following:
\begin{align}
W_{P\bm{I}\bm{O}}=\sum_{\substack{\alpha \vdash k\\H(\alpha)\leq d}} \sum_{k,l=1}^{d_{\alpha}}\frac{1}{m_{\alpha}}({E_{kl}^\alpha})_{\bm{I}}\otimes \tau_P \otimes ({E_{kl}^\alpha})_{\bm{O}}.
\end{align}
Now, exploiting properties of irreducible operator basis $\{E_{ij}^{\alpha}\}$ from expression~\eqref{irreps_basis_prop}, we write:
\begin{align}
\tr_{O}(W_{P\bm{I}\bm{O}})&=\sum_{\substack{\alpha \vdash k\\H(\alpha)\leq d}} \sum_{k,l=1}^{d_{\alpha}}\frac{\tr[({E_{kl}^\alpha})_{\bm{O}}]}{m_{\alpha}}({E_{kl}^\alpha})_{\bm{I}}\otimes \tau_P\\
&=\sum_{\substack{\alpha \vdash k\\H(\alpha)\leq d}} \sum_{k=1}^{d_{\alpha}} (E^{\alpha}_{kk})_{\bm{I}}\otimes \tau_P\\
&=\sum_{\substack{\alpha \vdash k\\H(\alpha)\leq d}} (P^{\alpha})_{\bm{I}}\otimes \tau_P\\
&=\id_{\bm{I}}\otimes \tau_P,
\end{align}
and finally $\tr(\tau_P)=1$. Having the above, we are in position to evaluate the minimal value of the constant $c$ satisfying constraint~\eqref{prob:dua2}.
\begin{lemma}
\label{dual:lemma1}
The minimal value of the constant $c$ satisfying constraint~\eqref{prob:dua2} for the operator $\hat{S}$ from~\eqref{eqn:ansatzW} is given by,
\begin{align}
\label{l:max_c_dual_homo}
    c= \frac{1}{d}\max_{\lambda\vdash (k+1)}\frac{\sum_{\beta\in\lambda}m_\beta}{m_\lambda},
\end{align}
where $m_{\lambda}, m_{\beta}$ are multiplicities of  irreps $\lambda\vdash k+1, \beta \vdash k$ in the Schur-Weyl duality. The symbol $\beta \in \lambda$ denotes all Young diagrams obtained from a Young diagram $\lambda$ by removing a single box.
 \end{lemma}
 
 \begin{proof}
 
 Writing the constraint from~\eqref{prob:dua2} explicitly by using expressions~\eqref{Omega_27}  and~\eqref{eqn:ansatzW} we get the following inequality
 \begin{align}
 \label{ineq}
   &\frac{1}{d^2}\sum_{\substack{\lambda\vdash (k+1)\\ H(\lambda)\leq d}}\sum_{p,q=1}^{d_{\lambda}}\frac{ (E_{pq}^{\lambda})_{\bm{I}P}\otimes (E_{pq}^{\lambda})_{\bm{O}F}}{m_{\lambda}}\\ & \hspace{5mm} \leq c \cdot  
   \sum_{\substack{\alpha \vdash (k)\\H(\alpha)\leq d}} \sum_{k,l=1}^{d_{\alpha}}\frac{d}{m_\alpha} ({E_{kl}^\alpha})_{\bm{I}} \otimes \tau_P\otimes ({E_{kl}^\alpha})_{\bm{O}}  \otimes \tau_F. \nonumber
\end{align}
To keep the notation simple, we drop indices the $\bm{I},\bm{O},P,F$, since it will be clear from the context which systems are referred to. First, we notice that the right-hand side of~\eqref{ineq} can be re-written by applying Lemma ~\ref{cor:lin_comb_of_E} with $\tau=\id/d$, leading to
\begin{equation}
\label{eq:aux1}
\begin{split}
&\sum_{\alpha} \sum_{k,l=1}^{d_{\alpha}}\frac{d}{m_\alpha} {E_{kl}^\alpha}\otimes \tau\otimes {E_{kl}^\alpha} \otimes \tau \\
&=\sum_{\alpha} \sum_{k,l=1}^{d_{\alpha}}\frac{1}{dm_\alpha} {E_{kl}^\alpha} \otimes \id\otimes {E_{kl}^\alpha}  \otimes \id\\
&=\sum_{\alpha} \sum_{k_\alpha,l_\alpha=1}^{d_{\alpha}} \sum_{\lambda,\lambda' \ni \alpha}\frac{1}{dm_\alpha}E^{\alpha\alpha}_{k_\alpha l_\alpha}(\lambda)\otimes E^{\alpha\alpha}_{k_\alpha l_\alpha}(\lambda')\\
&=\sum_{\lambda,\lambda'}\sum_{\alpha \in \lambda \wedge \lambda'}\sum_{k_\alpha,l_\alpha=1}^{d_{\alpha}}\frac{1}{dm_\alpha}E^{\alpha\alpha}_{k_\alpha l_\alpha}(\lambda)\otimes E^{\alpha\alpha}_{k_\alpha l_\alpha}(\lambda') .
\end{split}
\end{equation}
Here, we introduced symbol $\alpha \in \lambda \wedge \lambda'$, indicating that simultaneously $\alpha \in \lambda$ and $\alpha \in \lambda'$.
%Now, by applying Lemma~\ref{cor:lin_comb_of_E}, every irreducible basis operator labelled by $\alpha \vdash k$ extended by identity are block-diagonal with respect to irreps $\lambda\vdash k+1$, so that by \eqref{eqn:collary2} they have the following specific form
%\begin{align}
%\label{MtoD}
%    E^\alpha_{ij}\otimes \id = \sum_{\lambda\ni\alpha} E^{\alpha\alpha}_{i_\alpha j_\alpha}(\lambda).
%\end{align}
Next, by Notation~\ref{cor:lin_comb_of_E2} every irreducible operator basis labelled by $\lambda\vdash k+1$ on the left-hand side of equation~\eqref{ineq}  can be written as%in terms of the irreducible basis of $\beta \vdash k$ such that $\beta \in \lambda$, according to expression~\eqref{eqn:collary3}, giving us the following expression:
\begin{align}
\label{DtoM}
E_{pq}^{\lambda}=E^{\beta \beta'}_{p_{\beta}q_{\beta'}}(\lambda).
\end{align}
The above expression and final result of~\eqref{eq:aux1} allows us to re-write inequality~\eqref{ineq} in the following form:
\begin{align}
\label{ineqA}
&\frac{1}{d^2}\sum_{\lambda}\sum_{\beta,\beta'\in\lambda}\sum_{p_{\beta},q_{\beta'}}\frac{E^{\beta\beta'}_{p_{\beta}q_{\beta'}}(\lambda)\otimes E^{\beta\beta'}_{p_{\beta}q_{\beta'}}(\lambda)}{m_{\lambda}} \\
&\hspace{5mm}\leq c\cdot \sum_{\lambda,\lambda'}\sum_{\alpha \in \lambda \wedge \lambda'}\sum_{k_\alpha,l_\alpha=1}^{d_{\alpha}}\frac{1}{dm_\alpha}E^{\alpha\alpha}_{k_\alpha l_\alpha}(\lambda)\otimes E^{\alpha\alpha}_{k_\alpha l_\alpha}(\lambda'). \nonumber
\end{align}
Now, notice that every $\lambda,\lambda' \vdash (k+1)$ label a different nonequivalent irrep. This means that we deal with different orthogonal blocks. From the above inequality it is clear that non-trivial solutions for the constant $c$ can be obtained only for $\lambda=\lambda'$ -- for $\lambda \neq \lambda'$ we do not have cross terms on the left-hand side of~\eqref{ineqA}. Consequently, it is sufficient to restrict ourselves to comparing terms for every $\lambda$ separately:
%Combining~\eqref{MtoD},~\eqref{DtoM} with~\eqref{ineq}, we get for every $\lambda\vdash k+1$:
\begin{align}
\label{ineq2}
&\frac{1}{d^2}\sum_{\beta,\beta'\in\lambda}\sum_{p_{\beta},q_{\beta'}}\frac{E^{\beta\beta'}_{p_{\beta}q_{\beta'}}(\lambda)\otimes E^{\beta\beta'}_{p_{\beta}q_{\beta'}}(\lambda)}{m_{\lambda}} \\
&\leq c\cdot \sum_{\alpha\in\lambda}\sum_{i_{\alpha},j_{\alpha}}\frac{1}{dm_{\alpha}}E^{\alpha\alpha}_{i_\alpha j_\alpha}(\lambda) \otimes E^{\alpha\alpha}_{i_\alpha j_\alpha}(\lambda). \nonumber
\end{align}
All the operators in the above expression are the irreducible basis operators for the symmetric groups $S_{k+1}$ and $S_k$. Hence, by Schur-Weyl duality, they act non-trivially only on the symmetric component of the Schur-Weyl decomposition~\eqref{eq:SW}. 
Now, using the explicit forms of the irreducible basis operators in the Schur basis from~\eqref{eq:basisSWbasis} we rewrite~\eqref{ineq2} in the following way (ignoring the multiplicity space):
\begin{align}
&\frac{1}{d^2m_{\lambda}}\sum_{\beta,\beta'\in\lambda}\sum_{p_{\beta},q_{\beta'}}|\beta,p_{\beta}\>\<\beta',q_{\beta'}|\otimes |\beta,p_{\beta}\>\<\beta',q_{\beta'}| \\
&\leq \sum_{\alpha \in \lambda}\sum_{i_{\alpha},j_{\alpha}} \frac{c}{dm_{\alpha}}|\alpha,i_{\alpha}\>\<\alpha,j_{\alpha}|\otimes |\alpha,i_{\alpha}\>\<\alpha,j_{\alpha}|. \nonumber
\end{align}
Using the definition $|\Phi^+_{\alpha}\>=\sum_{i_{\alpha}=1}^{d_{\alpha}}|i_{\alpha}\>|i_{\alpha}\>$ for an unnormalised maximally entangled state on irrep $\alpha$ and rearranging the irreps, we get
\begin{align}
&\frac{1}{d^2m_{\lambda}}\sum_{\beta, \beta'\in \lambda}|\beta\beta\>\<\beta'\beta'|\otimes|\Phi^+_{\beta}\>\<\Phi^+_{\beta'}| \\
&\leq\sum_{\alpha \in \lambda} \frac{c}{dm_{\alpha}}|\alpha\alpha\>\<\alpha\alpha|\otimes |\Phi^+_{\alpha}\>\<\Phi^+_{\alpha}|. \nonumber
\end{align}
 Now, we notice that by applying an isometry $|\beta\beta\>\mapsto |\beta\>$ we can rewrite the above inequality as:
\begin{align}
&\frac{1}{d^2m_{\lambda}}\sum_{\beta, \beta'\in \lambda}|\beta\>\<\beta'|\otimes|\Phi^+_{\beta}\>\<\Phi^+_{\beta'}| \\
&\leq \sum_{\beta \in \lambda} \frac{c}{dm_{\beta}}|\beta\>\<\beta|\otimes |\Phi^+_{\beta}\>\<\Phi^+_{\beta}|.
\end{align}
 Defining a new set of normalised vectors $|\psi_{\beta}\>:=\frac{|\widetilde{\psi}_{\beta}\>}{||\widetilde{\psi}_{\beta}||}=\frac{1}{\sqrt{d_{\beta}}}|\widetilde{\psi}_{\beta}\>$, where $|\widetilde{\psi}_{\beta}\>:=|\beta\>\otimes |\Phi^+_{\beta}\>$, we rewrite the above equation as
\begin{align}
\label{ineq3}
 \frac{1}{d^2m_{\lambda}}\sum_{\beta,\beta'}\sqrt{d_{\beta}d_{\beta'}}|\psi_{\beta}\>\<\psi_{\beta'}|\leq \frac{c}{d}\sum_{\beta\in\lambda}\frac{d_{\beta}}{m_{\beta}}|\psi_{\beta}\>\<\psi_{\beta}|.
\end{align}
Notice that the matrix $A:=\sum_{\beta\in\lambda}\frac{d_{\beta}}{m_{\beta}}|\psi_{\beta}\>\<\psi_{\beta}|$ is diagonal and of the full rank. These properties allow us to compute easily $A^{-1/2}$, which is $A^{-1/2}=\sum_{\beta\in\lambda}\sqrt{\frac{m_{\beta}}{d_{\beta}}}|\psi_{\beta}\>\<\psi_{\beta}|$, since blocks labelled by different $\beta$ are orthogonal. Sandwiching inequality~\eqref{ineq3} by $A^{-1/2}$ and using fact that $\<\psi_{\beta}|\psi_{\beta'}\>=\delta_{\beta\beta'}$, gives the following after a few simple steps
\begin{align}
 &\frac{1}{dm_{\lambda}}\sum_{\beta,\beta'\in\lambda}\sqrt{m_{\beta}m_{\beta'}}|\psi_{\beta}\>\<\psi_{\beta'}|\leq c\cdot\id,\\
 &\frac{1}{dm_{\lambda}}|\psi\>\<\psi|\leq c\cdot \id \label{ineq4},
\end{align}
where we defined $|\psi\>:=\sum_{\beta\in\lambda}\sqrt{m_{\beta}}|\psi_{\beta}\>$, with $||\psi||^2=\sum_{\beta\in\lambda}m_{\beta}$. Expression~\eqref{ineq4} implies that 
\begin{align}
    &\forall \ \lambda \vdash k+1 \quad \frac{1}{d}\frac{\sum_{\beta\in\lambda}m_\beta}{m_\lambda}\leq c.
\end{align}
To satisfy the above inequality, one has to choose the constant $c$ at least equal to the maximal possible value of the ratio $(1/d)\sum_{\beta\in\lambda}m_\beta/m_{\lambda}$ calculated for all possible $\lambda'$s. This finally leads to
\begin{align}
\label{max_c_dual_homo}
    c = \frac{1}{d}\max_{\lambda\vdash (k+1)}\frac{\sum_{\beta\in\lambda}m_\beta}{m_\lambda}.
\end{align}
\end{proof}
Lemma~\ref{dual:lemma1} connects the minimal value of the constant $c$ from~\eqref{prob:dua2} with the specific choice of the operator $\hat{S}$ in~\eqref{eqn:ansatzW} and group theoretical parameters describing irreps of the symmetric groups $S_k$ and $S_{k+1}$. However, the final expression~\eqref{l:max_c_dual_homo} still involves a maximisation problem. 
To solve this maximisation explicitly, let us define an auxiliary function of Young diagrams of $n$ boxes:
\begin{equation}
\label{e:clambda}
    c(\lambda):=\frac{\sum_{\beta\in\lambda}m_\beta}{m_\lambda}.
\end{equation}
For the above defined function we can formulate the following:
\begin{proposition}
For a Young diagram $\lambda \vdash n$ with its step representation $(k_1,p_1,k_2,p_2,\ldots,k_s,p_s)$,  the function from~\eqref{e:clambda} is given by
\small
\begin{align}
\label{e:clambdap}
    &c(\lambda)=\sum_{j=1}^s\frac{k_j p_j}{d-\sum_{i=1}^jk_i+\sum_{i=j}^sp_i} \prod_{i=1}^{j-1} \left(1+\frac{k_i}{d(c_i,c_j)}\right) \\ 
    &\ \hspace{9mm} \prod_{i=j+1}^s \left(1+\frac{p_i}{d(c_j,c_i)}\right), \nonumber
\end{align}
\normalsize
where we assume that the product is equal to $1$ when the lower range of the products is larger than the upper one.
\end{proposition}
\begin{proof}
Let $(k_1,p_1,\ldots,k_s,p_s)$ be a step representation of $\lambda \vdash n$.
Note that removing a one box from the Young diagram $\lambda$ produces a new Young diagram if and only if the box removed is a corner box. Diagrams $\beta\in\lambda$ can therefore be labelled by $\mu=1,2,\ldots,s$. We use the convention that $\beta_\mu$ denotes the diagram obtained by removing the box $c_\mu$. According to formula \eqref{e:clambda}
\begin{equation}
\label{e:cl}
    c(\lambda)=\sum_{\mu=1}^s\frac{m_{\beta_\mu}}{m_\lambda}=
    \sum_{\mu=1}^s
    \frac{\prod_{b\in\beta_\mu}\dfrac{d-i_b+j_b}{h_\mu(b)}}{\prod_{b\in\lambda}\dfrac{d-i_b+j_b}{h(b)}},
\end{equation}
where $i_b,j_b$ denote the coordinates of a box $b$, $h_\mu(b)$ and $h(b)$ is the length of the hook of $b$ in $\beta_\mu$ and $\lambda$ respectively.  Here we used a relation between multiplicity of a given irrep with the number of semisimple Young diagrams~\cite{FultonSchur}.

For $\mu=1,2,\ldots,s$ we define the following subsets $R_\mu$, $C_\mu$ of $\lambda$ 
\begin{eqnarray}
    &R_\mu:=\{b\in\lambda:\,i_b=i_{c_\mu},\;j_b<j_{c_\mu}\},\\
    &C_\mu:=\{b\in\lambda:\,i_b<i_{c_\mu},\;j_b=j_{c_\mu}\}.
\end{eqnarray}
These subsets contain respectively boxes which are in the same row/column as the corner box $c_{\mu}$ not taking into account $c_\mu$ itself.
Let $G_\mu:=R_\mu\sqcup C_\mu\sqcup\{c_\mu\}$. Observe that $h_\mu(b)=h(b)$ if $b\not\in G_\mu$. Thus, \eqref{e:cl} simplifies to 
%\begin{equation} 
\begin{align}
    c(\lambda)=&
    \sum_{\mu=1}^s
    \frac{
    \prod_{b\in G_\mu\setminus\{c_\mu\}}\dfrac{d-i_b+j_b}{h_\mu(b)}
    }
    {
    \prod_{b\in G_\mu}\dfrac{d-i_b+j_b}{h(b)}
    } \\    &
    =
    \sum_{\mu=1}^s
    \frac{1}{d-i_{c_\mu}+j_{c_\mu}}
    \prod_{b\in C_\mu}
    \frac{h(b)}{h(b)-1}
    \prod_{b\in R_\mu}
    \frac{h(b)}{h(b)-1},
    \label{e:cle}
%\end{equation} 
\end{align}
where we used the fact that $h(c_\mu)=1$, and $h_\mu(b)= h(b)-1$ for $b\in G_\mu\setminus\{c_\mu\}$. We define 
\begin{equation}
C_{\mu,\nu}:=\{b\in C_\mu:\,i_{c_{\nu-1}}<i_b\leq i_{c_\nu}\}
\end{equation}
for $\nu=1,\ldots,\mu$, and
\begin{equation}
R_{\mu,\nu}:=\{b\in R_\mu:\,j_{c_{\nu+1}}<j_b\leq j_{c_\nu}\}
\end{equation}
for $\nu=\mu,\ldots,s$, with the convention that $i_{c_0}=0=j_{c_{s+1}}$. Observe that $C_\mu=\bigsqcup_{\nu=1}^\mu C_{\mu,\nu}$ and $R_\mu=\bigsqcup_{\nu=\mu}^s R_{\mu,\nu}$, and for $b\in C_\mu\sqcup R_\mu$,
\begin{equation}
h(b)=\begin{cases}
i_{c_\mu}-i_b+j_{c_\nu}-j_{c_\mu}+1, & \mbox{if $b\in C_{\mu,\nu}$,} \\
i_{c_\nu}-i_{c_\mu}+j_{c_\mu}-j_b+1, & \mbox{if $b\in R_{\mu,\nu}$.}
\end{cases}
\end{equation}
Hence, exploiting \eqref{e:cj}, we have
\small
\begin{align}
  &  \prod_{b\in C_\mu}\frac{h(b)}{h(b)-1}=
    \prod_{\nu=1}^\mu \prod_{b\in C_{\mu,\nu}}\frac{h(b)}{h(b)-1}\\
    &=\prod_{\nu=1}^\mu \prod_{b\in C_{\mu,\nu}}\frac{i_{c_\mu}-i_b+j_{c_\nu}-j_{c_\mu}+1}{i_{c_\mu}-i_b+j_{c_\nu}-j_{c_\mu}} \\
    &=\prod_{i=i_{c_{\mu-1}}+1}^{i_{c_\mu}-1}\frac{i_{c_\mu}-i+1}{i_{c_\mu}-i}
    %\times\\
    %&
    \prod_{\nu=1}^{\mu-1} \prod_{i=i_{c_{\nu-1}}+1}^{i_{c_\nu}} \frac{i_{c_\mu}-i+j_{c_\nu}-j_{c_\mu}+1}{i_{c_\mu}-i+j_{c_\nu}-j_{c_\mu}}
    \\
    &=(i_{c_\mu}-i_{c_{\mu-1}})\prod_{\nu=1}^{\mu-1}
    \frac{i_{c_\mu}-i_{c_{\nu-1}}+j_{c_\nu}-j_{c_\mu}}{i_{c_\mu}-i_{c_\nu}+j_{c_\nu}-j_{c_\mu}} \\
    &=k_\mu\prod_{\nu=1}^{\mu-1} \frac{k_\nu+p_\nu+\ldots+k_\mu}{p_\nu+\ldots+k_\mu} \\
    &=k_\mu\prod_{\nu=1}^{\mu-1}\left(1+\frac{k_\nu}{d(c_\nu,c_\mu)}\right).
\end{align}
\normalsize
Analogously, we have
\small
\begin{align}
    &\prod_{b\in R_\mu}\frac{h(b)}{h(b)-1}
    =\prod_{\nu=\mu}^s \prod_{b\in R_{\mu,\nu}}\frac{h(b)}{h(b)-1}\\
    &=\prod_{\nu=\mu}^s \prod_{b\in R_{\mu,\nu}}\frac{i_{c_\nu}-i_{c_\mu}+j_{c_\mu}-j_b+1}{i_{c_\nu}-i_{c_\mu}+j_{c_\mu}-j_b} \\
    &=\prod_{j=j_{c_{\mu+1}}+1}^{j_{c_\mu}-1}\frac{j_{c_\mu}-j+1}{j_{c_\mu}-j}
    %\times\\
    %&
    \prod_{\nu=\mu+1}^s \prod_{j=j_{c_{\nu+1}}+1}^{j_{c_\nu}} \frac{i_{c_\nu}-i_{c_\mu}+j_{c_\mu}-j+1}{i_{c_\nu}-i_{c_\mu}+j_{c_\mu}-j}
    \\
    &=(j_{c_\mu}-j_{c_{\mu+1}})\prod_{\nu=\mu+1}^s
    \frac{i_{c_\nu}-i_{c_{\mu}}+j_{c_\mu}-j_{c_{\nu+1}}}{i_{c_\nu}-i_{c_\mu}+j_{c_\mu}-j_{c_\nu}} \\
    &=p_\mu\prod_{\nu=\mu+1}^s \frac{p_\mu+\ldots+k_\nu+p_\nu}{p_\mu+\ldots+k_\nu} \\
    &=p_\mu\prod_{\nu=\mu+1}^s\left(1+\frac{p_\nu}{d(c_\mu,c_\nu)}\right).
\end{align}
Coming back to~\eqref{e:cle}, we get
\small
\begin{equation} %\begin{align}
    c(\lambda)=\sum_{\mu=1}^s
    \frac{k_\mu p_\mu}{d-i_{c_\mu}+j_{c_\mu}}\prod_{\nu=1}^{\mu-1}\left(1+\frac{k_\nu}{d(c_\nu,c_\mu)}\right)
    \prod_{\nu=\mu+1}^s\left(1+\frac{p_\nu}{d(c_\mu,c_\nu)}\right)
\end{equation} %\end{align}
\normalsize
what, after taking \eqref{e:cj} into account, ends the proof.
\end{proof}

As an example, let us consider the diagram $\lambda=[1,1,\ldots,1]=[1^{\times n}]$. Its step representation is $s=1$, $(k_1,p_1)=(n,1)$. Direct application of the formula \eqref{e:clambda} shows that
\begin{equation}
\label{e:asym}
    c([1^{\times n}])=\frac{n}{d-n+1}.
\end{equation}

Now, we are ready to prove the main theorem of this subsection
\begin{theorem}
    For every diagram $\lambda$ with $n$ boxes,
    \begin{equation}
        c(\lambda)\leq c([1^{\times n}]).
    \end{equation}
\end{theorem}
\begin{proof}
It was shown in \cite{vershik} that
\begin{equation}
\label{e:vershik}
    \sum_{j=1}^s k_j p_j \prod_{i=1}^{j-1} \left(1+\frac{k_i}{d(c_i,c_j)}\right) \prod_{i=j+1}^s \left(1+\frac{p_i}{d(c_j,c_i)}\right) = n.
\end{equation}
Observe that $\sum_{i=1}^jk_i\leq n$ and $\sum_{i=j}^sp_i\geq 1$ for every $j=1,2,\ldots,s$. Hence, $d-\sum_{i=1}^jk_i+\sum_{i=j}^np_i\geq d-n+1$. From \eqref{e:clambdap} and \eqref{e:vershik} we get
%\begin{equation}
\small
\begin{align}
    c(\lambda) \leq& \frac{1}{d-n+1}\sum_{j=1}^s k_j p_j \prod_{i=1}^{j-1} \left(1+\frac{k_i}{d(c_i,c_j)}\right)\prod_{i=j+1}^s \left(1+\frac{p_i}{d(c_j,c_i)}\right) \\
    =& \frac{n}{d-n+1}.
%\end{equation} 
\end{align}
According to \eqref{e:asym} the latter number is equal to $c([1^{\times n}])$.
\end{proof}
The above theorem clearly shows that the maximum possible value of the constant from~\eqref{max_c_dual_homo} is attained for antisymmetric partitions with $n=k+1$ boxes:
\begin{align}
c=\frac{1}{d}\max_{\lambda\vdash (k+1)} c(\lambda)=\frac{k+1}{d(d-k)}.
\end{align}
This exactly reproduces the lower bound obtained from the primal problem and proves optimally of chosen ansatz.

\section{Discussions}

We have presented  the optimal quantum circuits for transforming $k$ calls of a $d$-dimensional unitary operation $U_d$ into its complex conjugate $\overline{U_d}$ in a deterministic non-exact manner. 
When considering probabilistic exact transformations, Ref.~\cite{quintino19PRA} shows that transforming $k$ uses of $U_d$ into its complex conjugation $\overline{U_d}$ necessarily has probability one or probability zero, and shows that probability one can only be attained when $k\geq d-1$ -- completely solving the probabilistic exact complex conjugation. Hence, the problem of designing circuits for unitary complex conjugation is solved on the probabilistic exact paradigm and following our work now also on the deterministic non-exact one.  

A direct application of our circuit for unitary complex conjugation is to design a quantum circuit for reversing quantum operations.
In Ref.~\cite{quintino19PRA}, the authors present a simple probabilistic circuit which performs exact unitary transposition $U_d\mapsto U_d^T$ with probability $1/d^2$. 
Since $U_d^{-1}=\overline{U_d}^T$, if we perform a unitary transposition protocol after the complex conjugation circuit presented here, we can transform $k$ uses of $U_d$ into its inverse $U_d^{-1}$ with probability $1/d^2$ and an average fidelity of $\mean{F}=\frac{k+1}{d(d-k)}$.

As stated earlier, since $\overline{(U_dV_d)}=(\overline{U_d})\,(\overline{V_d})$, the complex conjugation function $f(U_d)=\overline{U_d}$ is a homomorphism. Also, complex conjugation is the only non-trivial homomorphism $f$ between $d$-dimensional special unitary operators $f:\SU(d)\to\SU(d)$. This follows from the fact that there exist only three $d-$dimensional representations for the group $\SU(d)$ \cite{harrisBook}: the trivial representation $f_\text{trivial}(U_d)=\id$, the defining representation, $f_\text{def}(U_d)=U_d$, and the conjugate representation $f_\text{conj}(U_d)=\overline{U_d}$. Hence, we have obtained optimal circuits for transforming $k$ uses of $U_d$ into $f(U_d)$ for all homomorphic functions $f:\SU(d)\to\SU(d)$.

\section*{Acknowledgements} MS, TM are supported through grant Sonata 16, UMO-2020/39/D/ST2/01234 from the Polish National Science Centre. 	M.T.Q. acknowledges the Austrian Science Fund (FWF) through the SFB project BeyondC (sub-project F7103), a grant from the Foundational Questions Institute (FQXi) as part of the  Quantum Information Structure of Spacetime (QISS) Project (qiss.fr). The opinions expressed in this publication are those of the authors and do not necessarily reflect the views of the John Templeton Foundation. This project has received funding from the European Union’s Horizon 2020 research and innovation programme under the Marie Skłodowska-Curie grant agreement No 801110.
It reflects only the authors' view, the EU Agency is not responsible for any use that may be made of the information it contains. ESQ has received funding from the Austrian Federal Ministry of Education, Science and Research (BMBWF). MH acknowledges support from the Foundation for Polish Science through IRAP project co-financed by the EU within the Smart Growth Operational Programme (contract no. 2018/MAB/5). For the purpose of Open Access, the author has applied a CC-BY public copyright licence to any Author Accepted Manuscript (AAM) version arising from this submission.

%%%%%%%%%%%%%%%%%%%%%%%%%%%%%%%%%%%%%%%%%%%%%%%%%%%%%%%%%%%%%%%%%%%%%%%
%%%%%%%%%%%%%%%%%%%%%%%% BIBLIOGRAPHY %%%%%%%%%%%%%%%%%%%%%%%%%%%%%%%%%
%%%%%%%%%%%%%%%%%%%%%%%%%%%%%%%%%%%%%%%%%%%%%%%%%%%%%%%%%%%%%%%%%%%%%%%
%\nocite{apsrev42Control} 
%\bibliographystyle{apsrev4-2.bst}
\bibliographystyle{0_MTQ_apsrev4-2_corrected}
\bibliography{0_MTQ_bib.bib}
%\begin{thebibliography}{30}

\end{document}